\documentclass[sigconf, nonacm]{acmart}
\usepackage{graphicx}
\usepackage[ruled, vlined, linesnumbered]{algorithm2e}
\usepackage{balance}
\usepackage{caption}
\usepackage{subcaption}
\usepackage{multirow}	
\usepackage{mathtools}
\usepackage{enumerate}
\usepackage{bbm}
\usepackage{url}
\usepackage{color,xcolor}
\usepackage{xspace}
\usepackage{afterpage}
\usepackage{enumitem}
\usepackage{titlesec}
\usepackage{relsize}

\hypersetup{colorlinks=true, citecolor=red, filecolor=blue, linkcolor=blue, urlcolor=blue, bookmarksdepth=3}

\renewcommand{\abstractname}{ABSTRACT}

\newlength{\oldtextfloatsep}
\newlength{\oldfloatsep}

\newcommand{\USim}{USim}

\newcommand{\RomanNum}[1]{\uppercase\expandafter{\romannumeral #1}}

\newcommand{\MaxSim}{\operatorname{MaxSim}}
\newcommand{\dis}{\operatorname{dis}}

\newcommand{\etal}{\emph{et~al.}\xspace}
\newcommand{\eg}{\emph{e.g.},\xspace}
\newcommand{\ie}{\emph{i.e.},\xspace}

\newcommand\figref[1]{Fig.~\ref{#1}}

\newcommand\tabref[1]{Table~\ref{#1}}

\newcommand\secref[1]{Sec.~\ref{#1}}
\newcommand\equref[1]{Eq.~(\ref{#1})}

\newcommand\algoref[1]{Algo.~\ref{#1}}

\newcommand\theref[1]{Theorem~\ref{#1}}

\newcommand\expref[1]{Example~\ref{#1}}
\newcommand\defref[1]{Definition~\ref{#1}}

\theoremstyle{definition}
\newtheorem{theorem}{Theorem}
\newtheorem{lemma}{Lemma}

\newtheorem{definition}{Definition}
\newtheorem{example}{Example}

\newcommand\vldbdoi{10.14778/3685800.3685895}
\newcommand\vldbpages{4441 - 4444}
\newcommand\vldbvolume{17}
\newcommand\vldbissue{12}
\newcommand\vldbyear{2026}
\newcommand\vldbauthors{Binhan Yang, Yuxiang Zeng, Hengxin Zhang, Zhuanglin Zheng, Yunzhen Chi, Yongxin Tong, Ke Xu}
\newcommand\vldbtitle{\shorttitle} 
\newcommand\vldbavailabilityurl{}
\newcommand\vldbpagestyle{empty} 

\newcommand{\fakeparagraph}[1]{\vspace{1mm}\noindent\textbf{#1.}}

\ifodd 0
\newcommand{\TODO}[1]{\textbf{\color{red}TODO:{ #1} }}

\newcommand{\ybh}[1]{{\color{blue}{#1}}}
\newcommand{\zeng}[1]{{\color{brown}{#1}}}
\else
\newcommand{\TODO}[1]{}

\newcommand{\zeng}[1]{#1}
\newcommand{\ybh}[1]{#1}
\fi

\newcommand{\MethodName}{MV-HNSW}
\newcommand\IndexName{{MV-HNSW}\xspace}

\begin{document}
\title{Unified and Efficient Approach for Multi-Vector Similarity Search}

\settopmatter{authorsperrow=1}

%
%
\author{Binhan Yang, Yuxiang Zeng, Hengxin Zhang, Zhuanglin Zheng, Yunzhen Chi, Yongxin Tong, Ke Xu}
\affiliation{%
  \institution{State Key Laboratory of Complex \& Critical Software Environment, Beihang University, Beijing, China}
    \institution{\{yangbh, yxzeng, zhanghengxin, zzlin, chiyz, yxtong, kexu\}@buaa.edu.cn}
}
\email{}

\begin{abstract}
Multi-Vector Similarity Search is essential for fine-grained semantic retrieval in many real-world applications, offering richer representations than traditional single-vector paradigms.  
Due to the lack of native multi-vector index, existing methods rely on a filter-and-refine framework built upon single-vector indexes.
By treating token vectors within each multi-vector object in isolation and ignoring their correlations, these methods face an inherent dilemma: aggressive filtering sacrifices recall, while conservative filtering incurs prohibitive computational cost during refinement. 
To address this limitation, we propose \IndexName{}, the first native hierarchical graph index designed for multi-vector data. 
\IndexName{} introduces a novel edge-weight function that satisfies essential properties (symmetry, cardinality robustness, and query consistency) for graph-based indexing, an accelerated multi-vector similarity computation algorithm, and an augmented search strategy that dynamically discovers topologically disconnected yet relevant candidates.
Extensive experiments on seven real-world datasets show that \IndexName{} achieves state-of-the-art search performance, maintaining over 90\% recall while reducing search latency by up to 14.0$\times$ compared to existing methods.
\end{abstract}

\maketitle

\pagestyle{\vldbpagestyle}
\begingroup\small\noindent\raggedright\textbf{PVLDB Reference Format:}\\
\vldbauthors. \vldbtitle. PVLDB, \vldbvolume(\vldbissue): \vldbpages, \vldbyear.\\
\href{https://doi.org/\vldbdoi}{doi:\vldbdoi}
\endgroup
\begingroup
\renewcommand\thefootnote{}\footnote{\noindent
This work is licensed under the Creative Commons BY-NC-ND 4.0 International License. Visit \url{https://creativecommons.org/licenses/by-nc-nd/4.0/} to view a copy of this license. For any use beyond those covered by this license, obtain permission by emailing \href{mailto:info@vldb.org}{info@vldb.org}. Copyright is held by the owner/author(s). Publication rights licensed to the VLDB Endowment. \\
\raggedright Proceedings of the VLDB Endowment, Vol. \vldbvolume, No. \vldbissue\ %
ISSN 2150-8097. \\
\href{https://doi.org/\vldbdoi}{doi:\vldbdoi} \\
}\addtocounter{footnote}{-1}\endgroup

\ifdefempty{\vldbavailabilityurl}{}{
\vspace{.3cm}
\begingroup\small\noindent\raggedright\textbf{PVLDB Artifact Availability:}\\
The source code, data, and/or other artifacts have been made available at \url{\vldbavailabilityurl}.
\endgroup
}

\section{INTRODUCTION}	\label{sec:introduction}

With the rapid advancement of Retrieval-Augmented Generation (RAG) \cite{DBLP:journals/corr/abs-2312-10997}, vector similarity search has become essential for enabling Large Language Models (LLMs) to access external knowledge. 
Given a query vector, vector similarity search aims to identify the top-$k$ most similar vectors from a large-scale collection, where similarity is typically measured by distance metrics such as inner product or Euclidean distance. 
This core operation underpins a wide range of applications, such as image retrieval, recommendation systems, and entity linking.

However, modern applications increasingly demand finer-grained semantic understanding. 
Conventional \textit{single-vector representation}, where each object is encoded as a single fixed-dimensional vector, often struggles to capture the rich details of complex unstructured data. 
For example, encoding an entire document as a single vector inevitably compresses local information, leading to semantic dilution \cite{khattab2020colbert, santhanam2022colbertv2}. 
To address this limitation, recent retrieval systems have shifted toward \textit{multi-vector representations}, where each object is decomposed into smaller units (\eg tokens or passages), each represented as a distinct vector (referred to as a ``\textit{\zeng{token} vector}'' throughout this paper). This paradigm, known as \textbf{M}ulti-\textbf{V}ector \textbf{S}imilarity \textbf{S}earch (\textbf{MVSS}) \cite{lee2023rethinking, khattab2020colbert}, preserves fine-grained semantics and enables more precise relevance matching through functions like $\operatorname{MaxSim}$ \cite{khattab2020colbert}. A typical example is as follows.

\begin{example}[Document Retrieval]
Consider an academic document retrieval system. 
A user submits a query: ``climate change impacts on biodiversity''. 
In this system, both queries and documents are represented as multi-vectors: the query is encoded into token-level vectors, while each document (\eg a research paper) is chunked into passages (\eg abstract, methodology, experiments), each encoded as a separate vector.
To measure relevance, the $\operatorname{MaxSim}$ function \cite{khattab2020colbert} is commonly used. 
For a given document, we examine each query token (\eg ``climate change'', ``biodiversity'') and identify its most similar passage within that document based on semantic relevance. 
These maximum similarities are then summed across all query tokens. 
A document is considered relevant if its passages collectively cover the query's semantics, with each query token finding a strong semantic match somewhere in the document.
\end{example}

\fakeparagraph{Prior Work on Multi-Vector Similarity Search}
Most existing approaches follow a filter-and-refine paradigm. 
In the filtering phase, heuristic strategies prune irrelevant objects based on similarities among \zeng{token} vectors. 
In the refinement phase, exact multi-vector similarity is computed on the remaining candidates.
For filtering, existing methods \cite{khattab2020colbert, lee2023rethinking, gao2021coil, li2023citadel} leverage single-vector indexes (\eg HNSW \cite{malkov2018efficient} and IVF \cite{sivic2003video}) or product quantization \cite{jegou2010product} to prune dissimilar objects. 
During refinement, they enumerate all \zeng{token} vector pairs between each candidate and the query to compute multi-vector similarity according to functions like $\operatorname{MaxSim}$.

\fakeparagraph{Limitations of Existing Work}
Current solutions suffer from critical limitations in either efficiency or recall. 
\textit{First}, their filtering strategies treat the \zeng{token} vectors within a multi-vector in isolation, building separate single-vector indexes. 
This ignores intra-object correlations and leads to avoidable recall loss.
\zeng{\textit{Second}, refinement phases typically compute expensive multi-vector similarity functions in a brute-force manner, creating an inherent dilemma: aggressive filtering to reduce computational cost risks severe recall degradation, while conservative filtering to preserve recall incurs prohibitive overhead during refinement.
These two fundamental limitations arise from a single structural gap: \textbf{the absence of native multi-vector indexing support} in existing vector databases and search engines, such as \ybh{Faiss \cite{douze2025faiss}, Milvus \cite{wang2021milvus} and Weaviate \cite{weaviate2025}}.
This situation forces most methods to adapt single-vector indexes for multi-vector data.
Thus, balancing recall and efficiency remains an open challenge for multi-vector similarity search.}

\fakeparagraph{Technical Challenges and Our Solution} 
Inspired by the success of the HNSW index \cite{malkov2018efficient} in single-vector similarity search, this paper introduces \IndexName{}, a native hierarchical graph index for multi-vector data. 
Adapting graph-based indexing to the multi-vector setting presents three technical challenges:
\begin{enumerate}
    \item \textbf{Unsuitable Multi-Vector Similarity Metric}: Multi-vector similarity functions like $\operatorname{MaxSim}$ lack essential properties required for graph edge weights, such as symmetry.
    
    \item \textbf{Computational Inefficiency}: Expensive multi-vector similarity computations hinder both index construction and similarity search.
    
    \item \textbf{Topological Disconnect}: Due to multi-vector structural complexity, true nearest neighbors may be topologically distant or disconnected from traversed nodes within the graph index.
\end{enumerate}




To address these challenges, \IndexName{} proposes an end-to-end framework integrating three synergistic components.
\textit{First}, we introduce a novel edge-weight function that averages the bidirectional similarity between two multi-vector objects and prove that it satisfies three essential properties (\ie symmetry, cardinality robustness, and query consistency) required for graph-based vector indexes.
\textit{Second}, we design an accelerated algorithm that leverages clustering to efficiently approximate common multi-vector similarity functions, significantly reducing computational overhead.
\textit{Third}, we develop an augmented search strategy that adaptively expands the candidate list during graph traversal, boosting recall while incurring minimal efficiency loss.

\fakeparagraph{Contribution} This paper makes the following contributions:
\begin{itemize}
    \item \textbf{Unified Problem Formalization}: We formalize the \textbf{U}nified \textbf{M}ulti-\textbf{V}ector \textbf{S}imilarity \textbf{S}earch problem, which unifies representative multi-vector similarity metrics (\eg MaxSim \cite{khattab2020colbert} and Aggregate $\gamma \operatorname{NN}$ \cite{lee2023rethinking}) as its special instantiations.
    
    \item \textbf{Essential Property Definition}: We formally define three core properties, \textit{symmetry}, \textit{cardinality robustness}, and \textit{query consistency}, that a valid edge-weight function must satisfy for multi-vector graph indexes. We show that directly adopting existing multi-vector similarity functions (\eg MaxSim) as edge weights in HNSW violates all three.

    \item \textbf{Native Multi-Vector Graph Index}: We propose \IndexName{}, the first hierarchical graph index natively designed for multi-vector data. We design a novel edge-weight function that satisfies all three essential properties.
    
    \item \textbf{Efficiency \& Recall Optimizations}: We present optimization methods to accelerate index construction and similarity search while improving recall.

    \item \textbf{Experiment Evaluation}: We conduct extensive experiments on 7 real-world datasets. Results demonstrate that our \IndexName{} significantly outperforms the state-of-the-art solutions \cite{khattab2020colbert, santhanam2022colbertv2, lee2023rethinking, scheerer2025warp} in search latency and recall.
\end{itemize}

\fakeparagraph{Road Map} 
The rest of this paper is organized as follows. \secref{sec:preliminaries} formalizes the problem. \secref{sec:3} and \secref{sec:SearchAugumentation} introduce our \IndexName{} index and search algorithm. 
\secref{sec:experimentalstudy}, \secref{sec:related-work}, and \secref{sec:CONCLUSION} present the experimental study, related work, and conclusion, respectively.

\newcommand{\ans}{\mathcal{A}}
\newcommand{\dist}{\operatorname{dis}}
\newcommand{\knn}{\gamma\operatorname{NN}}
\newcommand{\nn}{\operatorname{NN}}
\newcommand{\Data}{\mathcal{D}}
\newcommand{\CiteMaxSim}{khattab2020colbert, santhanam2022colbertv2, santhanam2022plaid, louis2025colbert, park2025scv, faysse2025colpali}
\newcommand{\CiteXTR}{lee2023rethinking,scheerer2025warp}

\begin{table}[h]
\centering
\captionsetup{skip=2.4pt}
\caption{Summary of major notations}\label{tab:notations} 
\begin{tabular}{cc}
\hline
Notation & Description \\
\hline
$v=(w_1, \cdots, w_d)$ & A $d$-dimensional vector data $v$ \\
$V=(v_1,\cdots,v_c)$ & A multi-vector $V$ with $c$ \zeng{token} vectors\\
$\mathcal D = \{V_1, \cdots, V_n\}$ & A dataset $\mathcal D$ of $n$ multi-vector data\\
\ybh{$\mathcal {V_{D}}$} & \zeng{A set of all token vectors in the dataset $\mathcal D$} \\
$Q$ & The query multi-vector data\\
$k$ & Number of neighbors needed for the query \\
$\knn(q, V)$ & $\gamma$ nearest neighbors of $q$ in $V$ \\
$\operatorname{\dist}(v_1, v_2)$ & Vector data distance between $v_1$ and $v_2$ \\
$\operatorname{\USim}(Q, V)$ & Unified multi-vector similarity from $Q$ to $V$ \\
\hline
\end{tabular}
\vspace{-3ex}
\end{table}

\section{\MakeUppercase{Problem Statement}}	\label{sec:preliminaries}

This section formally defines the studied problem.
The key notations used throughout the paper are summarized in \tabref{tab:notations}.

\subsection{Basic Concepts}

\begin{definition}[Vector Data]
A vector data object (``vector'' as short) is represented as a point $v = (w_1, w_2, \cdots, w_d)$ in the $d$-dimensional real coordinate space $\mathbb{R}^d$, where each $w_i \in \mathbb{R}$ denotes the $i$-th coordinate of $v$.
Given two such objects, their similarity can be quantified by a distance function $\dist: \mathbb{R}^d \times \mathbb{R}^d \to \mathbb{R}$.
\end{definition}

Common choices for the distance function $\dist(\cdot,\cdot)$ include Euclidean distance, inner product, and cosine similarity. In this work, we adopt the inner product $\langle \cdot, \cdot \rangle$ as the default distance function, following the convention that a larger distance indicates greater similarity. Our solution also supports other distance metrics.

\begin{definition}[Single-Vector Similarity Search] 
Given a (single) vector dataset $V$, a query vector $q \in \mathbb{R}^d$, and a positive integer $k$, the single-vector similarity search returns a subset $\ans$ of exactly $k$ vectors from the dataset $V$ that are nearest to $q$, \ie satisfying
\begin{equation}
    \forall u \in \ans, \forall v \in (V \setminus \ans), \dist(u, q) \ge \dist(v, q)
\end{equation}
We use $\knn(q,V)$ to denote these $\gamma$-nearest neighbors of $q$ in $V$.
\end{definition}

\begin{definition}[Multi-vector Data]
A multi-vector data object is denoted as a finite sequence of vectors, formally $V = (v_1, v_2, \cdots, v_{c})$ with $c > 1$, and each vector $v_i$ shares the same dimension $d$.
\zeng{We refer to each $v_i$ as a \textbf{token vector}, adapting the terminology of \textbf{token embeddings} from LLMs.}
\end{definition}

A multi-vector data provides a more fine-grained representation of unstructured data.
In practice, such a representation is typically created in two steps: (1) applying a \textit{segmentation strategy} to divide the raw data into multiple segments, and then (2) encoding each segment independently using an embedding model.
The cardinality $c$ corresponds to the number of segments, which in turn depends on the granularity of the segmentation strategy.

To illustrate this process, we use text data as an example and introduce two  segmentation strategies used in passage retrieval: 
\begin{itemize}
\item \textbf{Token-Level Segmentation}: Methods like \ybh{Col}BERT~\cite{khattab2020colbert} generate a contextualized embedding for each token in the text, so each vector encodes token-level semantics within its surrounding context. 

\item \textbf{Passage-Level Segmentation}: Methods like DPR~\cite{karpukhin2020dense} first split a document into fixed-size text windows (referred to as ``passages''), then embed each passage into a single vector.
\end{itemize}

\begin{definition}[Unified Multi-Vector Similarity Function]
For multi-vector data \ybh{$Q$} and \ybh{$V$}, the unified multi-vector similarity function $\operatorname{\USim}(Q, V)$ is defined as the \zeng{weighted} average distance from each token vector $q \in Q$ to its \ybh{$\gamma$}-nearest neighbors in $V$:
\begin{equation}\label{eq:USim}
    \operatorname{\USim}(Q, V) = \sum_{q \in Q} \left( w_q \cdot \frac{1}{\gamma} \sum_{v \in \knn(q, V)}{\dist(q, v)} \right) 
\end{equation}
\noindent\zeng{where $w_q \in [0, 1]$ is the input weight assigned to token vector $q$.}
\end{definition}

Compared to single-vector similarity evaluation, the unified function $\operatorname{\USim}$ for multi-vectors is more complex and computationally intensive.
Moreover, several multi-vector similarity metrics proposed in prior work can be expressed as special cases of $\operatorname{\USim}$. Below we introduce \zeng{three} representative examples:
\begin{itemize}
\item \textbf{MaxSim} is the most prevalent multi-vector similarity metric. It corresponds to the case where $\gamma = w_q = 1$ in \equref{eq:USim}. 
It was first introduced in ColBERT~\cite{khattab2020colbert} and has since been used in many subsequent studies~\cite{santhanam2022colbertv2, santhanam2022plaid, li2023slim, shrestha2024espn, nardini2024efficient, louis2025colbert, park2025scv, faysse2025colpali}.

\item \zeng{\textbf{Weighted Chamfer} was proposed to explicitly account for importance of different token vectors in multi-vector similarity search~\cite{garg2025incorporating}. 
It corresponds to the instantiation of \equref{eq:USim} with $\gamma=1$, where $w_q$ are learnable weights.}

\item \textbf{Aggregate $\knn$} was originally introduced in Google's framework (called XTR~\cite{lee2023rethinking}) for multi-vector similarity search.
It corresponds to the instantiation of \equref{eq:USim} with \ybh{$\gamma>1$, $w_q = 1$}, and $\dist$ defined as the inner product.
This metric has also been adopted in follow-up work \cite{scheerer2025warp}.
\end{itemize}

\vspace{-1ex}
\begin{example} \label{exp:maxsim}
To illustrate the computation of $\operatorname{\USim}$ with $\gamma = 1$, consider a multi-vector $V = \{v_1, v_2, v_3\}$ and a query $Q = \{q_1, q_2, q_3\}$. 
Their token vectors and associated weights are given as follows:
\begin{itemize}[noitemsep]
    \item $Q$: $q_1=(1,0)$, $q_2=(0,1)$, and $q_3=(\frac{1}{\sqrt 2}, \frac{1}{\sqrt 2})$.
    \item $V$: $v_1=(0.8, 0.6)$, $v_2=(0.6, 0.8)$, and $v_3=(\frac{1}{\sqrt 2}, \frac{1}{\sqrt 2})$.
    \item Weight: $w_{q_1} = 1, w_{q_2} = 0$ and $w_{q_3} = 1$.
\end{itemize}
Based on \equref{eq:USim} the similarity from $Q$ to $V$ is computed by taking, for each token vector $q \in Q$, the maximum inner product with any token vector in $V$, and then aggregating the weighted results:
\begin{equation*}
\begin{aligned}
\mathrm{USim}(Q, V) &= \sum_{q_i \in Q} w_{q_i} \cdot \max_{v_j \in V} \; \langle q_i, v_j \rangle \\
&= 1 \cdot \max\Bigl\{ 0.8,\; 0.6,\; \tfrac{1}{\sqrt{2}} \Bigr\}
   + 0 \cdot \max\Bigl\{ 0.6,\; 0.8,\; \tfrac{1}{\sqrt{2}} \Bigr\} + \\
&\quad 1 \cdot \max\Bigl\{ \tfrac{0.8+0.6}{\sqrt{2}},\; 
                    \tfrac{0.6+0.8}{\sqrt{2}},\; 
                    1 \Bigr\} = 0.8+0+1=1.8
\end{aligned}
\end{equation*}
\end{example}

\vspace{-1ex}
\subsection{Problem Definition}
We formally define the studied problem as follows.

\vspace{-1ex}
\begin{definition}[\textbf{U}nified \textbf{M}ulti-\textbf{V}ector \textbf{S}imilarity \textbf{S}earch]
Given a multi-vector dataset $\Data$, a query multi-vector $Q$ \zeng{with associated token weights $w_q$}, and a target integer $k$, the unified multi-vector similarity search aims to retrieve a subset $\mathcal{A} \subseteq \Data$ \zeng{consisting of the $k$ multi-vectors most similar to $Q$ according to the unified multi-vector similarity function $\operatorname{\USim}$}.
Formally, the search result $\mathcal{A}$ must satisfy the following conditions:
\begin{enumerate}
\item $|\mathcal{A}| = k$ and $\mathcal{A} \subseteq \Data$;
\item $\forall U \in \mathcal{A}$ and $V \in (\Data \setminus \mathcal{A}),  \operatorname{\USim}(Q, U) \geq \operatorname{\USim}(Q, V)$.
\end{enumerate}
\vspace{-1ex}
\end{definition}

\expref{exp:UMSSP} illustrates a concrete instance of the problem.

\vspace{-2ex}
\begin{example}\label{exp:UMSSP}
Consider a multi-vector dataset $\mathcal{D} = \{V_1, V_2, V_3\}$ and a query multi-vector $Q = \{q_1, q_2\}$. Their token vectors are as follows:
\begin{itemize}[noitemsep]
    \item $V_1 = \{v_1, v_2\}$: $v_1 = (\frac{\sqrt{3}}{2}, \frac{1}{2}, 0)$, $v_2 = (0, \frac{4}{5}, \frac{3}{5})$.
    \item $V_2 = \{v_3, v_4\}$: $v_3 = (\frac{1}{\sqrt{2}}, \frac{1}{\sqrt{2}}, 0)$, $v_4 = (0, \frac{3}{5}, \frac{4}{5})$.
    \item $V_3 = \{v_5, v_6\}$: $v_5 = (\frac{3}{5}, \frac{4}{5}, 0)$, $v_6 = (0, 1, 0)$.
    \item $Q = \{q_1, q_2\}$: $q_1 = (1, 0, 0)$, $q_2 = (0, \frac{1}{\sqrt{2}}, \frac{1}{\sqrt{2}})$.
\end{itemize}
Assume $\gamma = 1$, $w_{q_1} = w_{q_2} = 1$, and $k = 2$. 
Computing $\operatorname{\USim}$ from the query multi-vector $Q$ to each multi-vector $V_i \in \mathcal{D}$ yields:
\begin{equation*}
\begin{aligned}
\operatorname{USim}(Q, V_1) 
&= \langle q_1, v_1 \rangle + \langle q_2, v_2 \rangle 
= \tfrac{\sqrt{3}}{2} + \tfrac{7\sqrt{2}}{10} \approx 1.856, \\
\operatorname{USim}(Q, V_2) 
&= \langle q_1, v_3 \rangle + \langle q_2, v_4 \rangle 
= \tfrac{1}{\sqrt{2}} + \tfrac{7\sqrt{2}}{10} \approx 1.697, \\
\operatorname{USim}(Q, V_3) 
&= \langle q_1, v_5 \rangle + \langle q_2, v_6 \rangle 
= \tfrac{3}{5} + \tfrac{1}{\sqrt{2}} \approx 1.307.
\end{aligned}
\end{equation*}

With $k=2$, the unified multi-vector similarity search returns $\{V_1, V_2\}$ as the top-2 results with the highest similarity scores.
\end{example}
\vspace{-1ex}

\fakeparagraph{Remark}
Due to the curse of dimensionality \ybh{\cite{indyk1998approximate}} and the expensive computational cost of evaluating multi-vector similarity, exact solutions to the above problem do not scale to large-scale datasets.
Consequently, existing research has focused on approximate algorithms that trade some accuracy for efficiency, aiming to \textit{maximize the recall} as defined in \equref{equ:recall}:
\begin{equation} \label{equ:recall}
    \text{Recall} = \frac{|\mathcal{A} \cap \mathcal{A}^*|}{k}
\end{equation}
where $\mathcal{A}^*$ denotes to the exact answer.
Following this research direction, our work also focuses on the approximate algorithm for \zeng{unified} multi-vector similarity search.




\SetKwFunction{SearchLayer}{SEARCH-LAYER}

\vspace{-1ex}
\section{\MakeUppercase{Multi-Vector Index} \MakeUppercase{\IndexName}}\label{sec:3}

\zeng{This section elaborates on our index {\IndexName}, the first native hierarchical graph index for multi-vectors. We start with its main idea (\secref{sec:overview}), then introduce its structure and novel edge-weight function (\secref{sec:index-structure}), present the index construction algorithm (\secref{sec:index-construct}), and finally propose an optimized {\USim} computation method for accelerating offline construction and online search (\secref{subsec:EUC}).}

\vspace{-1ex}
\subsection{Main Idea}\label{sec:overview}

The Hierarchical Navigable Small World (HNSW) graph \cite{malkov2018efficient} has become a dominant index for single-vector similarity search. Its success stems from a multi-layered graph design that enables approximate $k$NN search in sublinear time, achieving state-of-the-art recall and latency trade-offs \cite{wang2021comprehensive}, making it one of the most popular index choices in modern vector databases. However, HNSW is designed to index individual token vectors, not multi-vectors.

Consequently, existing filter-and-refine based methods for multi-vector similarity search lack native multi-vector indexing support, suffering from either low recall or high search latency. To address these limitations, we propose a multi-vector index ({\IndexName}) and an efficient solution for unified multi-vector similarity search. Our framework consists of two synergistic components:

\textbf{(1) Multi-Vector Graph Index}.
We design a hierarchical graph index, where each node represents a multi-vector. An edge connects two nodes if their corresponding multi-vectors are similar enough according to a similarity fusion function.
To accelerate index construction, we introduce an approximate method for efficiently computing the {\USim} score.

\textbf{(2) Augmented Search via Hierarchical Graphs}.
Using {\IndexName}, we can progressively find the $k \operatorname{NN}$s to the query multi-vector, but navigating this graph alone still misses relevant answers.
To mitigate this, we expand the candidate set for every explored node during traversal by incorporating globally relevant multi-vectors into its connected neighbors (see \secref{sec:SearchAugumentation} for detail).





\subsection{Multi-Vector Index Structure}\label{sec:index-structure}

We introduce our \IndexName index in two parts: the core design principle (\secref{sec:index-structure-idea}) followed by its detailed structure (\secref{sec:index-structure-detail}).

\subsubsection{Core Design Principle.}\label{sec:index-structure-idea}
To extend HNSW's success to multi-vector similarity search, we propose \IndexName{}, a novel HNSW-based index. 
In this index, each node represents a multi-vector, and edges connect similar multi-vectors. 
The \textbf{core technical challenge} lies in defining a suitable edge-weight function $f(u, v)$ between nodes $u$ and $v$. 
We first identify and formalize three essential properties that such a function $f(\cdot,\cdot)$ must possess: \textit{symmetry}, \textit{cardinality robustness}, and \textit{query consistency}.
We then demonstrate that the unified multi-vector similarity function $\operatorname{\USim}$ violates these criteria, justifying the need for our subsequent solution in \secref{sec:index-structure-detail}.

\fakeparagraph{Essential Property for Edge Weight Function}
We now formally define the three essential properties for the edge weight function.

\begin{definition}[Symmetry]
    An edge weight function $f$ is symmetric if $f(u, v) = f(v, u)$ for any two multi-vectors $u$ and $v$.
\end{definition}

This weight symmetry property ensures that traversing graph is consistent: if you can go from nodes from $u$ to $v$, then $v$ is equally reachable from $u$. This ensures consistent navigation.

\begin{definition}[Cardinality Robustness]
    An edge weight function $f$ is cardinality-robust if it does not exhibit systematic bias towards larger cardinalities of multi-vectors.
    That is, $f(u, v)$ should not systematically increase or decrease with $|u|$ or $|v|$.
\end{definition}

Without this property, some multi-vectors would attract an excessive number of edges due to their high cardinality rather than genuine similarity.
This isolates closely relevant nodes and creates hub nodes that monopolize connections, thereby harming recall.

\begin{definition}[Query Consistency]
    An edge weight function $f$ is said to be query-consistent if, whenever two nodes are similar under $f$, their $\operatorname{\USim}$ scores remain similar for any query.
     Formally, for any small $\tau \geq 0$, if \ybh{$f(u,v) \geq 1 - \tau$}, then 
     \begin{equation}
        \left| \operatorname{\USim}(Q, u) - \operatorname{\USim}(Q, v) \right| = O(\tau \cdot |Q|).
     \end{equation}
\end{definition}

This property guarantees that when the search moves from a node $u$ to its connected neighbor $v$ with high weights, their similarity to the query remains comparable.
Without this property, following a local edge could lead to a node with a drastically different {\USim} score, making the traversal unstable and inefficient.

\fakeparagraph{Why {\USim} Fails as an Edge Weight Function}
We now demonstrate that the unified multi-vector similarity function {\USim}, while a natural candidate for $f$, fails to satisfy the essential properties defined above.

\begin{lemma}
    {\USim} is asymmetric, not cardinality-robust, and not query-consistent.
\end{lemma}
\begin{proof}

    \textbf{(1) Violation of Symmetry.}
    The definition in \equref{eq:USim} is inherently asymmetric. For $\gamma = 1$, $\operatorname{USim}(u, v)$ computes the weighted sum of distance from each \zeng{token vector} $u_i \in u$ to its nearest neighbor in $v$, whereas $\operatorname{USim}(v, u)$ computes that from each $v_j \in v$ to its nearest neighbor in $u$. Since the nearest-neighbor relation is not symmetric, {\USim} is asymmetric.
    
    \textbf{(2) Violation of Cardinality Robustness.}
    From \equref{eq:USim}, $\operatorname{\USim}(u, v)$ sums over all $u_i \in u$. Thus, for fixed $v$ and $\gamma = 1$, $\operatorname{\USim(u, v)} = \Theta(|u|)$ when the nearest neighbor distances are bounded. This linear dependence on $|u|$ directly violates cardinality robustness.

    \textbf{(3) Violation of Query Consistency.}
    \zeng{Query consistency does not hold in general. Consider the following counterexample.
    Let $D_1 = (v_1, v_2)$ and $D_2 = (v_1, v_2, v_3, v_4)$ be multi-vector datasets, and let $Q = (v_3, v_4)$ be a query multi-vector with weight $w_{v_3} = w_{v_4} = 1$, where each $v_i$ denotes a token vector.
    We assume both token vectors $v_3,v_4$ are orthogonal to every token vector in $D_1$, which can be trivially instantiated in high-dimensional space.
    Then, when $\gamma = 1$, we have $\operatorname{USim}(Q, D_1) = 0$ and $\operatorname{USim}(Q, D_2) = 2$.
    Then $\frac{\operatorname{USim}(D_1, D_2)}{|D_1|} = 1$, giving $\tau = 0$, but $|\operatorname{USim}(Q, D_1) - \operatorname{USim}(Q, D_2)| = |Q| \neq 0$. Hence, the query consistency is violated.
    }
\end{proof}

\subsubsection{Index Structure.}\label{sec:index-structure-detail}

\begin{figure*}[t]
    \centering
    \includegraphics[width=\textwidth]{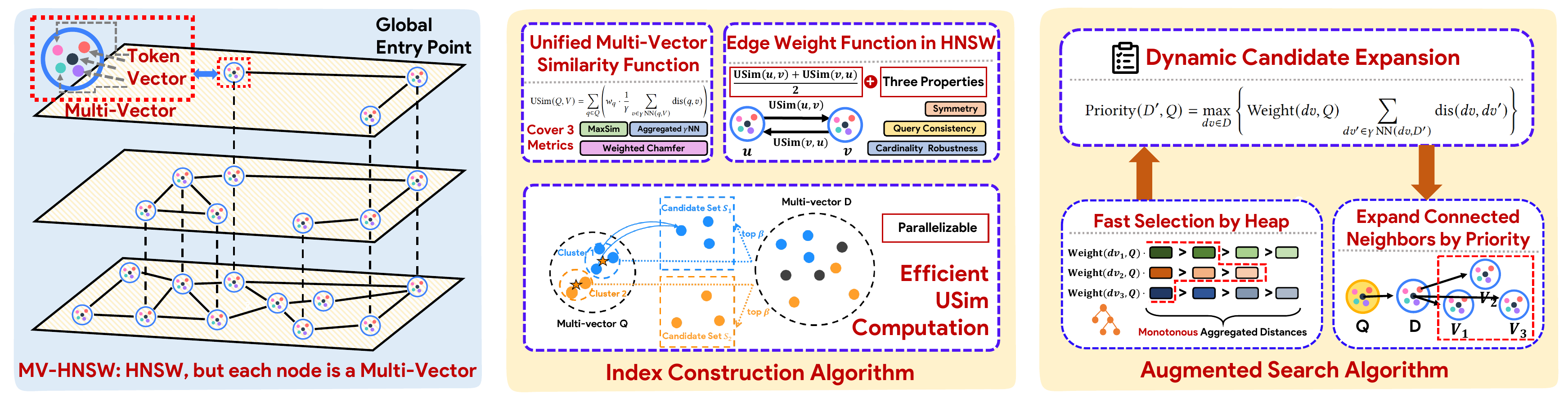}
    \vspace{-4ex}
    \caption{Overview of our solution {\IndexName}}\label{fig:dco-idea}
    \label{fig:framework}
    \vspace{-3ex}
\end{figure*}

As shown in \ybh{\figref{fig:framework}}, \IndexName is a hierarchical graph with multiple layers, where each node represents a unique multi-vector,
Within a layer, nodes are connected by undirected edges that indicate the similarity between the corresponding multi-vectors, and each node maintains at most $M$ such connections to its (approximate) nearest neighbors.

When constructing this index, each multi-vector $D_i$ from the dataset $\Data$ is first assigned a random layer $l_i$ sampled from an exponentially decaying distribution. 
\zeng{A node representing $D_i$ is then inserted into all layers from $l_i$ down to the base layer.
Accordingly, the highest layer contains the fewest nodes and provides the global entry point $ep$ for similarity search.}
This design ensures two key properties: (1) higher layers contain progressively sparser subsets of the dataset, which enables rapid coarse-grained navigation, and (2) similarity search can proceed top-down, starting from the coarsest layer and refining results through finer layers.

\fakeparagraph{Edge Weight Function}
The key adaptation that enables effective multi-vector search lies in the design of edge weight function $f$.
As aforementioned, a suitable $f$ must satisfy three properties: \textit{symmetry}, \textit{cardinality robustness}, and \textit{query consistency}.
To meet these requirements, we introduce a novel similarity fusion function (defined in \defref{def:edge}) to quantify edge weights between multi-vectors, and provide a formal proof in \theref{the:edge}.

\begin{definition}[Edge Weight Function]\label{def:edge}
    The edge weight function $f(u,v)$ is defined as the average of bidirectional normalized similarity between two multi-vectors $u$ and $v$:
    \begin{equation} \label{eq:mvsim}
        f(u,v) = \frac{1}{2}\left(\frac{\operatorname{\USim}(u, v)}{|u|} + \frac{\operatorname{\USim}(v, u)}{|v|}\right)
    \end{equation}
    \noindent\zeng{where token weights used in $\operatorname{\USim}$ follow the same pre-processing method as query \zeng{token} weights, defaulting to $1$ if unspecified.}
\end{definition}

\begin{theorem}\label{the:edge}
    The edge weight function $f$ defined in \equref{eq:mvsim} is symmetric, cardinality-robust, and query-consistent.
\end{theorem}
\begin{proof}
    The theorem follows directly from Lemmas \ref{lem:symmetry}--\ref{lem:consistency}.
\end{proof}

\begin{lemma}\label{lem:symmetry}
    The edge weight function $f$ is symmetric.
\end{lemma}
\begin{proof}
    Based on the definition of $f$ in \equref{eq:mvsim}, we have
    \begin{equation}
       f(u,v) = \frac{1}{2}\left( \frac{\operatorname{\USim}(v, u)}{|v|} + \frac{\operatorname{\USim}(u, v)}{|u|} \right) = f(v,u)
    \end{equation}
\end{proof}
    
\begin{lemma}\label{lem:cardinality}
    The edge weight function $f$ is cardinality-robust.
\end{lemma}
\begin{proof}
    For any multi-vectors $u$ and $v$, the term $\frac{\operatorname{\USim}(u, v)}{|u|}$ represents the average similarity from \zeng{token} vectors of $u$ to $v$, which lies in $[-1, 1]$. Similarly, $\frac{\operatorname{\USim}(v, u)}{|v|} \in [-1, 1]$. 
    Since $f(u, v)$ is the average of these two values in $[-1, 1]$, we have $f(u, v) \in [-1, 1]$. Moreover, each term is normalized by the cardinality of its first argument, ensuring that $f(u,v)$ does not systematically increase with $|u|$ or $|v|$. This satisfies the definition of cardinality robustness.
\end{proof}


\begin{lemma}\label{lem:consistency}
    The edge weight function $f$ is query-consistent.
\end{lemma}
\begin{proof}[Proof Sketch]
Due to page limitations, we provide a proof sketch here, and more detailed proof can be found in our online full paper \cite{mvhnsw-2026}.
We first prove that when the edge weight $f(V_1, V_2)$ is sufficiently high, the expected number of token vectors in $V_1$ \zeng{that are highly similar to token vectors in $V_2$} is bounded. 
Applying Markov's inequality \cite{DBLP:books/daglib/0012859} then yields a constraint on the distance discrepancy for any query $Q$. 
By algebraically bounding the minimum distance gap, we deduce that
$\vert\operatorname{\USim}(Q, V_1) - \operatorname{\USim}(Q, V_2) \vert \leq O(\tau^z \vert Q \vert)$ \ybh{with $z \leq 1$ holds strictly}, which satisfies the definition of query consistency.
\end{proof}

\setlength{\textfloatsep}{2ex}
\setlength{\floatsep}{2ex}
\begin{algorithm}[t]
	\caption{Insert Multi-Vector $D_i$ into \IndexName{}}\label{alg:insertion}
	\KwIn{Node degree $M$, size $efC$ of construction candidate list, normalization factor $m_L$ for layer sampling}
	\KwOut{\IndexName{} after inserting the new multi-vector $D_i$}
    $l_i \gets \lfloor -\ln(\operatorname{uniform}(0, 1)) \cdot m_L \rfloor$\;
    $ep \gets$ the global entry point at layer $l_{ep}$\;
     \ForEach{layer $l_c \gets l_{ep}$ \KwTo $l_i + 1$} {
        $ep \gets$ \SearchLayer{$D_i, l_c, ep$} with $ef=1$\;
     }
     \For{$l_c \gets \min(l_i, l_{ep})$ \KwTo $0$} {
        $Cand \gets $ \SearchLayer{$D_i, l_c, ep$} with $ef=efC$\;
        $S \gets $ $M$ nearest candidates to $D_i$ from $Cand$\;
        Connect $D_i$ with each $D_{j} \in S$ at layer $l_c$\;
        \ForEach{multi-vector $D_j \in S$} {
            \If{$\operatorname{degree}(D_j) > M$} {
                Trim number of connected neighbor of $D_j$ to $M$\;
            }
        }
        $ep \gets Cand$\;
     }
     \If{$l_i$ is higher than current global entry point} {
        Change global entry point to $D_i$\;
     }
\end{algorithm}
\afterpage{\global\setlength{\textfloatsep}{\oldtextfloatsep}}
\afterpage{\global\setlength{\floatsep}{\oldfloatsep}}

\subsection{Index Construction Algorithm}\label{sec:index-construct}

\subsubsection{Main Idea.}
\zeng{The \IndexName index is constructed by sequentially inserting each multi-vector into a hierarchical graph. 
For each insertion, the algorithm first randomly assigns a layer $l_i$ to the new node, navigates from the topmost layer down to $l_i$ to obtain an entry point $ep$, and then inserts the node from layer $l_i$ down to the base layer. At each layer in this downward traversal, the insertion consists of two steps:}

\textbf{Navigation}: Starting from the global entry point, the algorithm invokes the \SearchLayer routine of HNSW \cite{malkov2018efficient} to perform a greedy search at the current layer and locate the nearest neighbor(s) of the new multi-vector.
This neighbor serves as the entry point(s) for the next (lower) layer.

\textbf{Connection}: During the traversal in current layer, the algorithm also retrieves $efC$ candidate neighbors, selects up to $M$ closest ones, and establishes bidirectional edges. 
\zeng{The degree of each node is maintained below $M$ through pruning when necessary.}

\subsubsection{Algorithm Details.} \label{sec:algodetails}
\algoref{alg:insertion} illustrates the detailed steps:

\zeng{\textbf{Localization} (Lines 1-4).
We first randomly assign a layer $l_i$ to the new multi-vector $D_i$.}
Unlike standard HNSW, our \SearchLayer uses \ybh{the edge weight function  as the multi-vector similarity metric.
The resulting node $ep$ serves as the entry point for the next layer, and this process repeats until the target layer $l_i$ is reached.}

\zeng{\textbf{Insertion} (Lines 5-12). 
Starting from the assigned layer $l_i$ down to the base layer, the algorithm inserts the new multi-vector $D_i$ iteratively. 
At each layer, the \SearchLayer routine retrieves a set $Cand$ of candidates, from which the top-$M$ nearest neighbors are selected to establish bidirectional connections with $D_i$. 
If a neighbor $D_j$'s degree exceeds $M$ after connection, its neighbor list is pruned to retain only the top-$M$ closest neighbors. 
The nearest neighbor in $Cand$ also serves as the entry point for the next lower layer.}

\zeng{If the assigned layer $l_i$ of $D_i$ exceeds that of the current global entry point, $D_i$ becomes the new global entry point (Lines 13-14).
This ensures that the entry point always resides at the highest layer.}

\subsubsection{Complexity Analysis.} \label{sec:construction-complexity}
\zeng{The complexity of \algoref{alg:insertion} is derived from the complexity analysis of standard HNSW \cite{malkov2018efficient}.}


\textbf{Time Complexity.} 
HNSW constructs an index for $n$ token vectors in $O(n\log{n})$ time.
By contrast, our construction algorithm introduces one additional cost per layer: computing $\operatorname{\USim}$ for multi-vector pairs during traversal.
Each (exact) $\operatorname{\USim}$ computation takes $O(dc^2\log{\gamma})$ time.
As $d, c$ and $\gamma$ are constants, the overall construction time remains $O(n\log{n})$, but with a larger constant factor.


\textbf{Space Complexity.}
HNSW stores $n$ token vectors using $O(n)$ space. 
Our {\IndexName} stores multi-vectors instead, where each node contains $c$ token vectors of dimension $d$, increasing per-node storage by a factor of $c$. Since $d$ and $c$ are constants, the overall space complexity remains $O(n)$.

\subsection{Accelerating {\USim} Computation} \label{subsec:EUC}

The computational cost of \USim, as revealed by our complexity analysis, is the main bottleneck for building {\IndexName}. Consequently, we target its acceleration to reduce total construction time.

\setlength{\textfloatsep}{2ex}
\setlength{\floatsep}{2ex}
\begin{algorithm}[t]
	\caption{Accelerate {\USim} Computation}\label{alg:acmaxsim}
	\KwIn{Two multi-vectors $D$ and $Q$ associated with weighted $w_q$, parameter $\beta$}
	\KwOut{Approximate result $ans$}
    Clusters $C_{1}, \cdots, C_{\sqrt{|Q|}} \gets$ $k$-means over $\{qv_i \in Q\}$\; 
    \ForEach(){$i$-th cluster with centroid $C_i$} {
        Initialize an empty $\beta$-sized heap $S_i$ keyed by distance\;
        \ForEach{token vector $d_j \in D$} {
            $dist \gets \dis(C_i, d_j)$\;
            \ybh{\If{$\vert S_i \vert < \beta$ \textbf{or} key of top \zeng{tuple} in $S_i < dist$} {
                Remove top \zeng{tuple} if $\vert S_i \vert = \beta$\;
                Insert $(dist, d_j)$ into $S_i$\;
            }}
        }
    }
    \For{cluster $i \gets 1$ \KwTo $\alpha$} {
        \ForEach{token vector $qv_j \in$ cluster $C_i$} {
            $tdis \gets $ sum of distance of $\knn$ to $qv_j$ from $S_i$;
            $ans \gets ans + \frac{w_{qv_j}}{\gamma}  \cdot tdis$;
        }
    }
    \KwRet{ans}\;
\end{algorithm}
\afterpage{\global\setlength{\textfloatsep}{\oldtextfloatsep}}
\afterpage{\global\setlength{\floatsep}{\oldfloatsep}}

\subsubsection{Main Idea: Filter-and-Refine with Clustering.}
We accelerate the computation of $\operatorname{\USim}(Q, D)$ using a \textit{filter-and-refine} strategy coupled with \textit{clustering} to balance efficiency and accuracy. This method operates at two complementary granularities:

\textbf{(1) Cluster Level Filtering}: 
The \zeng{token} vectors in the query $Q$ are partitioned into $\sqrt{\vert Q \vert}$ clusters via k-means clustering.
This $\sqrt{\vert Q \vert}$ choice balances the trade-off between granularity and efficiency, as it grows sublinearly with query sizes and is empirically effective.
For each cluster centroid, we perform $\beta$NN search over token vectors in $D$ to collect a candidate set of size \ybh{$\beta = \max\{\gamma, \sqrt{\vert D \vert}\}$}. 

\textbf{(2) Token-Vector Level Refinement}: 
Each \zeng{token} vector $qv_i \in Q$ then searches only within its cluster's candidate set to find its $\gamma$NN.
Since this step operates only on the pruned candidate set of size $\beta$ (rather than all $|D|$ token vectors), it retains high efficiency while preserving accuracy with negligible degradation.

\subsubsection{Algorithm Details.}
\algoref{alg:acmaxsim} presents the detailed algorithm. 
Line 1 clusters the token vectors $qv_i \in Q$ into $\sqrt{|Q|}$ centroids via k-means clusering \cite{DBLP:books/mk/HanKP2011}. For \textit{filtering} (Lines 2--8), we perform a $\beta$NN search for each centroid $C_i$ over all token vectors in $D$. Finally, in the \textit{refinement} step (Lines 9--12), we identify the $\gamma$NN for each query token vector $qv_i \in Q$ by scanning the pruned candidate sets of its cluster and compute the approximate $\operatorname{\USim}$ value.


\fakeparagraph{Complexity Analysis}
Since the number of clusters for high-dimensional vectors is typically set to \(\sqrt{n}\), we obtain \(\sqrt{|Q|}\) clusters in Line 1 in \(O(|Q| \sqrt{|Q|})\) time.  
To balance the time cost between the filtering and refinement steps, \(\beta\) is set to \(\sqrt{|D|}\).  
Let \(c = \max(|Q|, |D|)\) be the maximum number of token vectors in a multi-vector. The overall time cost of \algoref{alg:acmaxsim} is then \(O(c \sqrt{c})\).

\fakeparagraph{Parallelization} 
The filter-and-refine framework of this optimization algorithm lends itself naturally to parallel execution.
In the filtering phase, the $\beta$NN search for all $\sqrt{|Q|}$ centroids are independent and can thus be executed in parallel.
Similarly, in the refinement phase, \zeng{token} vectors from different clusters can be processed concurrently.
Overall, the pipeline enables $\sqrt{|Q|}$-way parallelism.

\newcommand{\hitSum}{\text{Contrib}}
\newcommand{\priority}{\text{Priority}}
\newcommand{\Wt}{Weight}

\section{\MakeUppercase{Augmented Search Algorithm for Multi-Vector Similarity Search}}	\label{sec:SearchAugumentation}

This section introduces our augmented search algorithm using \IndexName.
We begin with the main idea in \secref{sec:search-idea}, then detail the algorithm in \secref{sec:search-detail}, and finally analyze its complexity in \secref{sec:search-complexity}.

\subsection{Main Idea}\label{sec:search-idea}

\subsubsection{Motivation.}
Using {\IndexName}, a straightforward solution to unified multi-vector similarity search is to traverse this index from top to bottom while progressively collecting $k$ nearest candidates to the query $Q$, following the search method of HNSW \cite{malkov2018efficient}.

However, due to the dynamic nature of queries and the finite degree constraints in \IndexName{}, a critical issue arises: \textit{multi-vectors in the exact result may be topologically distant from, or even disconnected from, the traversed nodes}. 
This is because the index is constructed based on multi-vector similarity among data objects, but the search procedure relies on similarity between data objects and the query. 
This mismatch can prevent the search from reaching relevant candidates, leading to impaired recall and causing the algorithm to become trapped in local optima.




\subsubsection{Key Idea.}
To address this issue, we enhance the graph traversal by dynamically expanding the candidate list.
While exhaustively scanning all multi-vectors would achieve this goal in principle, it is computationally intractable for large-scale datasets.
\zeng{Thus, we propose a \textit{local-to-global} inference strategy to infer promising global candidates using locally identified candidates.
This strategy starts from local neighborhoods in the graph index and leverages the pre-computed correlations to reach distant but relevant multi-vectors that might otherwise remain undiscovered.}
The core idea is twofold: 

\textbf{(1) Offline Pre-Processing}: 
We pre-compute for each \zeng{token} vector a list of multi-vectors containing \zeng{token} vectors correlated with it in the dataset, stored in an Auxiliary Navigation Table (ANT). 

\textbf{(2) Online Search}: 
When exploring a multi-vector $V$ from the dynamic candidate list, we expand the search scope by leveraging ANT for multi-vectors containing high correlated \zeng{token} vectors of $V$'s \zeng{token} vectors (and thus similar to those in $Q$), bridging structural gaps in the graph index.


\subsection{Augmented Search Algorithm}\label{sec:search-detail}

\subsubsection{Core Components.}

To realize the above idea, we introduce two technical components: 
\textit{auxiliary navigation table} and \textit{dynamic candidate expansion strategy}.


\textbf{(1) Auxiliary Navigation Table (ANT).}
For each token vector $v$ in every multi-vector $V$ in the dataset $\mathcal{D}$, we pre-compute an $M$-sized list of relevant multi-vectors to $v$, denoted by $\mathcal{ANT}(v)$, where $M$ is the node degree in \IndexName. 
To construct this list, we first use the standard HNSW index \cite{malkov2018efficient} to retrieve a set $A$ of $M'$ nearest token vectors to $v$ from all token vectors, excluding those within $V$.
Here, \(M'\) is set to a constant multiple of \(M\) (\eg \(5M\) in our implementation) to ensure sufficient candidates after removal.
We then consider the multi-vectors that contain at least one token vector in $A$, and for each such multi-vector, we compute the sum of distances between \(v\) and its top-\(\gamma\) nearest token vectors in \(A\).
Finally, we retain the top-\(M\) multi-vectors based on this sum.

\textbf{(2) Dynamic Candidate Expansion Strategy.}
During online search, after exploring the connected neighbors of a candidate multi-vector $V$, we invoke dynamic candidate expansions as follows. 
For each \zeng{token} vector $v \in V$, we collect all the multi-vectors in $\mathcal{ANT}(v)$ via ANT and form the expansion candidate set $\mathcal{E} = \bigcup_{v \in V} \mathcal{ANT}(v)$.
This set has up to $|V| \cdot M$ distinct multi-vectors, so computing the $\operatorname{\USim}$ scores between these multi-vectors and the query $Q$ is costly.
Instead, we quantify the priority of each multi-vector $V \in \mathcal{E}$ based on \equref{eq:hitsum}–\equref{equ:priority} and only select the top-$M$ multi-vectors to be expanded from $\mathcal{E}$. 

\ybh{\vspace{-3ex}}
\begin{align}
    \hitSum(v, Q) &= \sum_{q \in Q} \mathbb{I}\big(v \in \gamma\operatorname{NN}(q, V)\big) \cdot w_{q} \cdot \dis(q, v) \label{eq:hitsum} \\
    \operatorname{\Wt}(v, Q) &= \frac{\exp \left( \hitSum(v, Q) \right)}{\sum_{v' \in V} \exp \left( \hitSum(v', Q) \right)} \label{eq:wt} \\
    \priority(V',Q) &= \max_{v \in V} \Bigg\{ \operatorname{\Wt}(v, Q) \sum_{v' \in \gamma \operatorname{NN}(v,V')} \dist(v, v') \Bigg\}
    \label{equ:priority}
\end{align}

\zeng{Specifically, $\hitSum(v, Q)$ defined in \equref{eq:hitsum} measures the contribution of a token vector $v \in V$ to the unified multi-vector similarity $\operatorname{\USim}(Q,V)$, where $\mathbb{I}(\cdot)$ is an indicator function. 
$\operatorname{\Wt}(v, Q)$ in \equref{eq:wt} is derived using softmax normalization to amplify high-contribution token vectors and suppress low-impact ones.
Based on this, $\priority(V',Q)$ quantifies the priority score between the candidate multi-vector $V'$ and the query $Q$.
To avoid the expensive computation of $\operatorname{\USim}(Q, V')$, $\priority(V',Q)$ estimates the similarity between $V'$ and $Q$ by considering, among all token vectors $v \in V$, the one that best matches $V'$ (via its $\gamma \operatorname{NN}$ in $V'$), weighted by $v$'s relevance to $Q$.
We adopt this max-based aggregation for two reasons: (1) to accelerate priority computation, and (2) to avoid confining the online search to graph-connected neighbors of $V$ and uncover topologically distant yet closely relevant candidates.}

\subsubsection{Algorithm Details.}

\newcommand{\ExN}{ExN}
\setlength{\textfloatsep}{2ex}
\setlength{\floatsep}{2ex}
\begin{algorithm}[t] 
	\caption{Augmented \textsf{SEARCH-LAYER}} \label{algo:augumentedsearchlayer}
	\KwIn{A query multi-vector $Q$, {\IndexName} index, ANT, entry points $ep$ at current layer $l_c$, search candidate list with size $efS$}
	\KwOut{$efS$ most similar multi-vectors to $Q$ at layer $l_c$}
    $Cand \gets ep, Queue \gets ep$\;
    \While{$Queue$ is not empty} {
        $(V_c,\dis_c) \gets$ most dissimilar node to $Q$ in $Cand$\;
        $(V_q,\dis_q) \gets$ most similar node to $Q$ in $Queue$, pop it\;
        \lIf{$\dis_c > \dis_q$}{\textbf{break}}
        \ybh{
        $\ExN \gets $ direct neighbors of $V_q$ in layer $l_c$\;
        \tcp{Dynamic Expansion Strategy}
        Max-heap $heap \gets \emptyset$ \tcp*[r]{Key: $\priority(\cdot, Q)$}
        \ForEach{token vector $v \in V_q$} {
            $V' \gets$ the first multi-vector in $\mathcal{ANT}(v)$\;
            $heap \gets heap \cup \{(\priority(V',Q), V', v)\}$\;
        }
        \ForEach{iteration $i \gets 1$ \KwTo $M$} {
            $(V, v) \gets$ pop the top tuple from $heap$\;
            $\ExN \gets \ExN \cup \{V\}$\;
            $V' \gets$ the next multi-vector in $\mathcal{ANT}(v)$\;
            $heap \gets heap \cup \{(\priority(V',Q), V', v)\}$\;
        }
        }
        \ForEach{unvisited multi-vector $V \in \ExN$} {
            Mark $V$ as visited, $\dis \gets \operatorname{\USim}(Q, V)$\;
            $(V_c,\dis_c) \gets$ most dissimilar node to $Q$ in $Cand$\;
            \If{$|Cand| < efS$ \textbf{or} $\dis > \dis_c$} {
                Insert $(V, \dis)$ into $Cand$, keep $|Cand| \le efS$\;
                Push $(V, \dis)$ into $Queue$\;
            }
        }
    }
    \KwRet{$Cand$}\;
\end{algorithm}
\afterpage{\global\setlength{\textfloatsep}{\oldtextfloatsep}}
\afterpage{\global\setlength{\floatsep}{\oldfloatsep}}

            
\zeng{\algoref{algo:augumentedsearchlayer} optimizes the \SearchLayer algorithm in the HNSW index \cite{malkov2018efficient} by incorporating our \textit{dynamic expansion strategy}.
The search procedure initializes $Cand$ as the candidate answer set and $Queue$ as the priority queue during traversal, both with the entry points $ep$ (Line 1).
$Cand$ maintains the top-$efS$ most similar results found so far, while $Queue$ stores nodes pending expansion, ordered by their similarity to the query $Q$.
At each iteration, let $V_c \in Cand$ be the candidate most dissimilar to $Q$ and $V_q \in Queue$ the multi-vector most similar to $Q$ so far. 
If the worst candidate $V_c$ is more similar to $Q$ than the best node $V_q$ awaiting for expansion, the search loop terminates (Line 5); otherwise, the algorithm proceeds to expand from $V_q$.} 

The set of neighbors $\ExN$ expanded from $V_q$ consists of two subsets: (1) the direct neighbors of $V_q$ in {\IndexName} (Line 6), and (2) those generated by the dynamic expansion strategy (Lines 7--15). 
For the latter, each token vector $v \in V_q$ contributes the first multi‑vector in its associated list $\mathcal{ANT}(v)$ as an initial candidate. This candidate $V'$ is inserted into a max‑heap keyed by $\priority(V', Q)$.
Lines 11--14 then repeatedly extract the top tuple from this heap and append it to $\ExN$, the set of nodes awaiting expansion.
This extraction repeats $M$ times, so $\ExN$ contains at most $2M$ multi-vectors, where $M$ is the node degree.

Finally, Lines 16--21 process each unvisited node $V \in \ExN$.
If this node $V$ is more similar to $Q$ than any node in $Cand$ (or if $|Cand| < efS$), then $V$ is inserted into both $Cand$ and $Queue$, and $Cand$ is trimmed to maintain a size of at most $efS$.

\fakeparagraph{Remark}
The dynamic expansion strategy is designed to select $M$ multi‑vectors from the union set $\mathcal{E} = \bigcup_{v \in V_q} \mathcal{ANT}(v)$ according to their $\priority(\cdot, Q)$ defined in \equref{equ:priority}. Rather than exhaustively scanning the entire set $\mathcal{E}$, the algorithm adaptively retrieves one multi‑vector from each $\mathcal{ANT}(v)$ list at a time in Lines 11–15.

This optimization is effective for two reasons. First, $\operatorname{\Wt}(v, Q)$ in \equref{equ:priority} remains constant for a given token vector $v$ across all candidate multi‑vectors in $\mathcal{ANT}(v)$.
Second, each list $\mathcal{ANT}(v)$ is pre‑ordered by the remaining term $\sum_{v' \in \gamma\operatorname{NN}(v, V')} \dist(v, v')$.
Consequently, for any list $\mathcal{ANT}(v)$, the front multi‑vector always has a higher priority than any subsequent multi‑vector in the same list. 
This monotonic property ensures that Lines 8--15 yield the top‑$M$ multi‑vectors from $\mathcal{E}$ without requiring a full scanning of $\mathcal{E}$.

\subsection{Complexity Analysis}\label{sec:search-complexity}
We mainly analyze the time complexity of our search algorithm, as the ANT requires only \(O(nc \cdot M) \sim O(n)\) space. The time complexity is examined in two phases: \textit{offline pre-processing} and \textit{online search}.

\textbf{Offline Pre-Processing.} 
The time overhead for constructing the ANT consists of two main steps: (1) building a single-vector HNSW index over all the token vectors $\mathcal{V_D}$, and (2) executing a $k$NN search on this index for every token vector. This leads to an overall offline pre-processing complexity of $\mathcal{O}(|\mathcal{V_D}| \log |\mathcal{V_D}|)$.

\textbf{Online Search.} 
\algoref{algo:augumentedsearchlayer} optimizes the standard \SearchLayer procedure by adapting the dynamic candidate expansion strategy. 
This algorithm additionally performs $O(M + |V_q|)$ heap operations for each explored node $V_q$. As $M$ and $|V_q|$ are constants in practice, each online search requires $O(\log{n})$ time.

\newcommand{\lifestyle}{Lifestyle}
\newcommand{\pooled}{Pooled}
\newcommand{\recreation}{Recreation}
\newcommand{\writing}{Writing}
\newcommand{\technology}{Technology}
\newcommand{\science}{Science}
\newcommand{\msmacro}{Macro}

\newcommand{\linearscan}{Linear Scan}
\newcommand{\colbert}{ColBERT}
\newcommand{\colbertvii}{ColBERTv2}
\newcommand{\xtr}{XTR}
\newcommand{\warp}{WARP}

\section{\MakeUppercase{Experimental Study}}
\label{sec:experimentalstudy}


\subsection{Experimental Setup}

\begin{table}[t]
\centering
\captionsetup{skip=2.4pt}
\caption{Statistics of datasets (M: Million, B: Billion)}\label{tab:datasets} 


\begin{tabular}{cccc}
\hline
\textbf{Dataset} & \textbf{Dim.} & \textbf{\#(Token Vector)} & \textbf{\#(Query)}\\
\hline
\recreation & 128--768 & 263.0 M & 563\\
\lifestyle & 128--768 & 268.9 M & 417\\
\writing & 128--768 & 277.1 M & 497\\
\science & 128--768 & 343.6 M & 538 \\
\technology{} & 128--768 & 127.6 M & 916 \\
\pooled & 128--768 & 2.429 B & 2.931 K \\
\msmacro & 128--768 & 8.842 B & 101.1 K \\
\hline
\end{tabular}
\end{table}

\fakeparagraph{Datasets} 
We evaluate our method on seven multi-vector datasets derived from two real-world benchmark corpora: LoTTE \cite{santhanam2022colbertv2} and MS MARCO \cite{nguyen2016ms}. 
As summarized in \tabref{tab:datasets}, these datasets cover a broad range of scales, from 263 million to nearly 9 billion token vectors. 
Each dataset provides token-level embeddings generated by the representation model of \colbertvii \cite{santhanam2022colbertv2} in four dimensionalities (128, 256, 512, and 768).
Each dataset also includes a set of queries of varying scales, which are similarly processed into token-level embeddings (\ie multi-vectors).
\zeng{To establish the ground truth, we perform a brute-force linear scan to compute the exact $\operatorname{USim}$ and rank the results to identify the exact answer.}
Notice that the above data preparation pipeline is commonly adopted by existing research on multi-vector similarity search \cite{khattab2020colbert, gao2021coil, santhanam2022colbertv2, santhanam2022plaid, lee2023rethinking, li2023citadel, scheerer2025warp}.
All datasets use the MaxSim multi-vector similarity function.

\begin{figure*}[t]
    \centering
    \begin{subfigure}{0.23\textwidth}
        \centering
        \includegraphics[width=\textwidth]{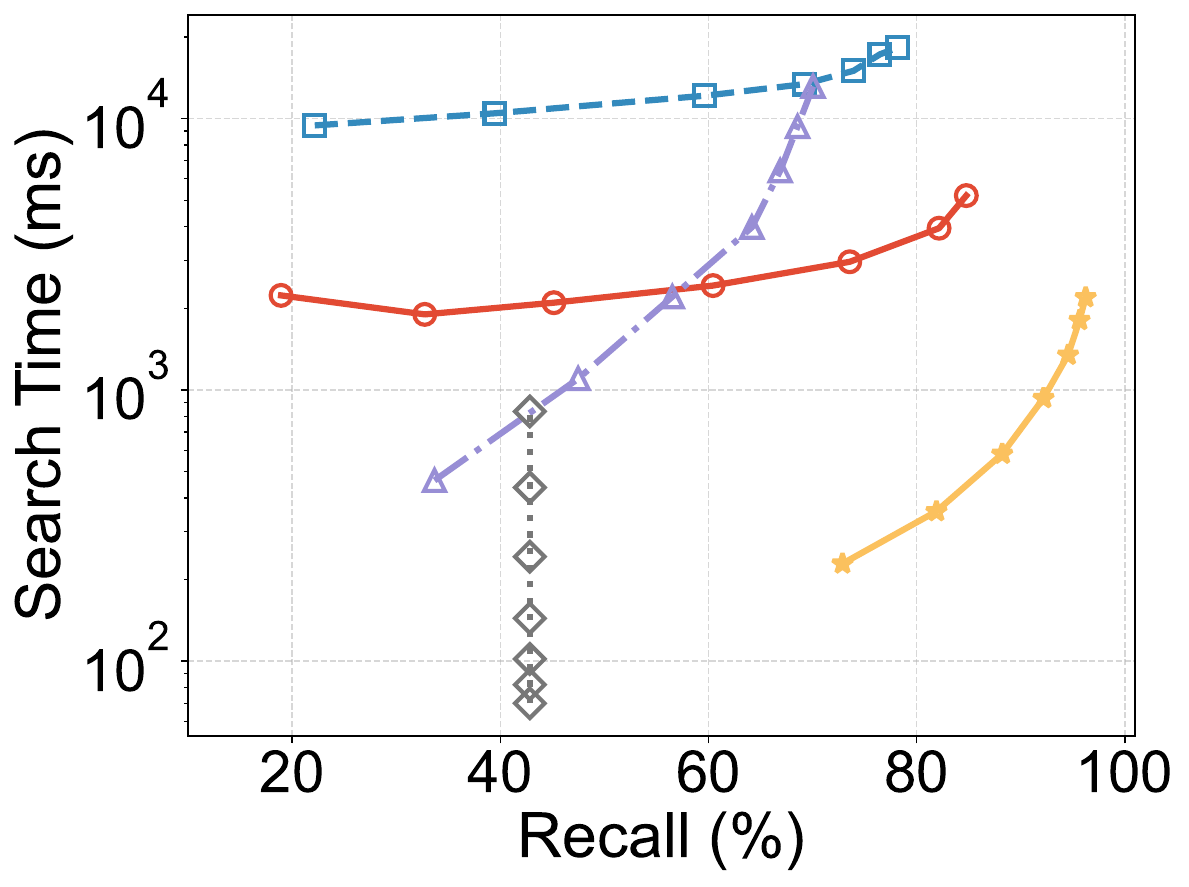}\vspace{-1.0ex}
        \caption{\recreation}
        \label{fig:recreation-default}
    \end{subfigure}
    \hfill
    \begin{subfigure}{0.23\textwidth}
        \centering
        \includegraphics[width=\textwidth]{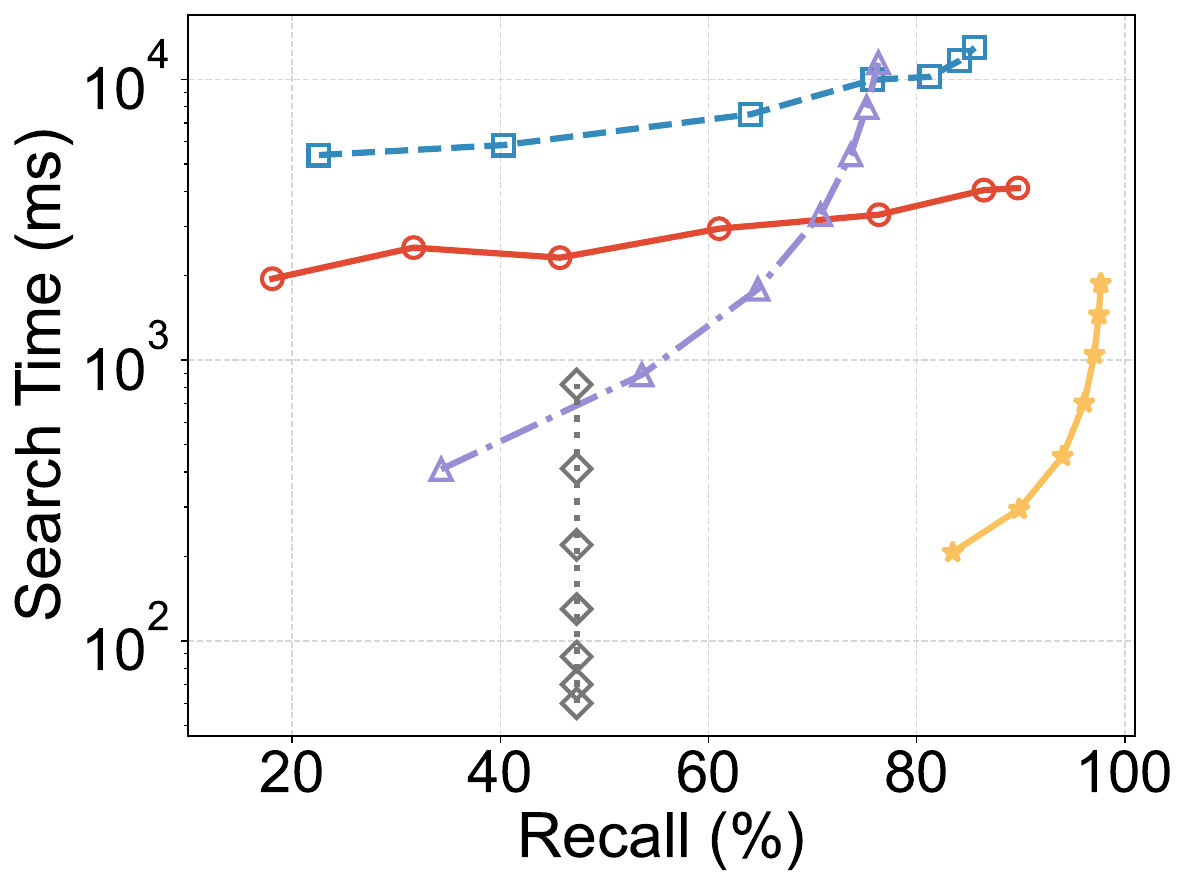}\vspace{-1.0ex}
        \caption{\lifestyle}
        \label{fig:lifestyle-default}
    \end{subfigure}
    \hfill
    \begin{subfigure}{0.23\textwidth}
        \centering
        \includegraphics[width=\textwidth]{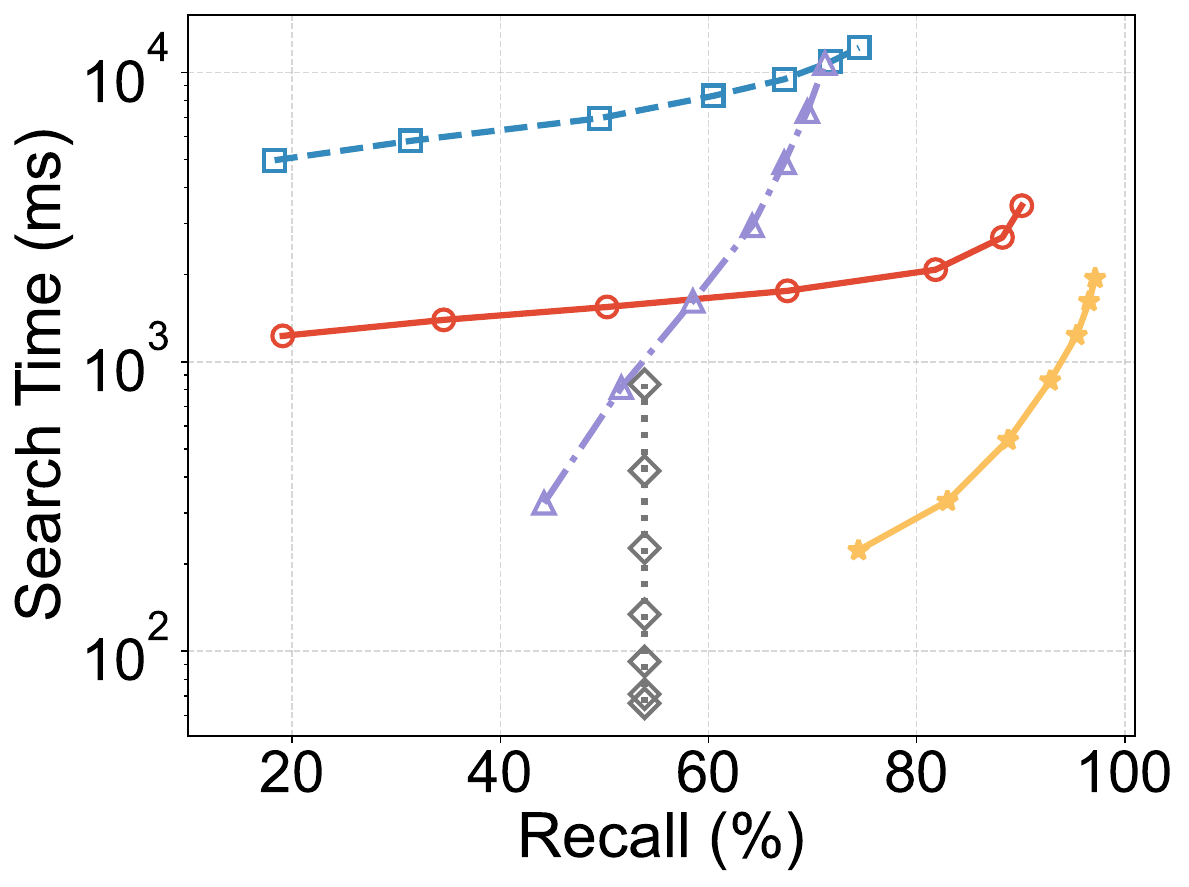}\vspace{-1.0ex}
        \caption{\writing}
        \label{fig:writing-default}
    \end{subfigure}
    \hfill
    \begin{subfigure}{0.23\textwidth}
        \centering
        \includegraphics[width=\textwidth]{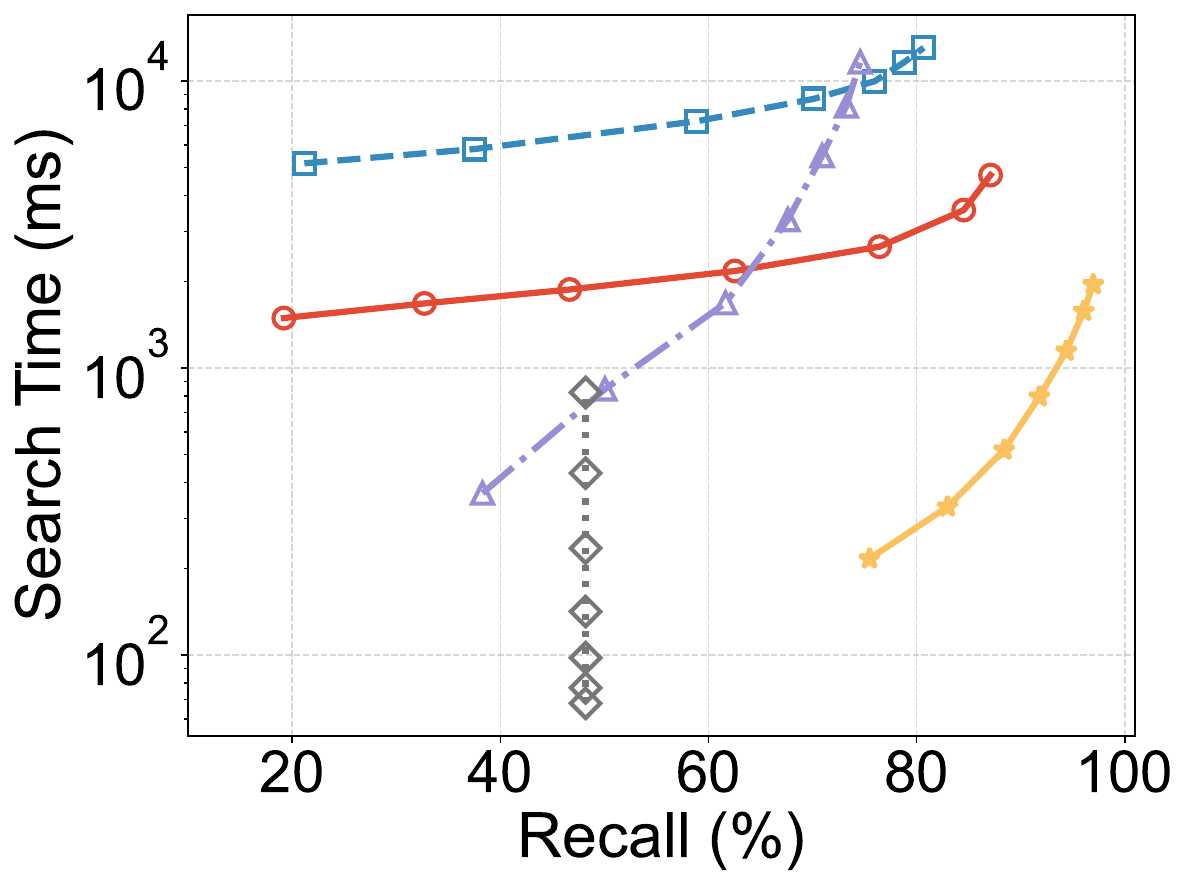}\vspace{-1.0ex}
        \caption{\lifestyle}
        \label{fig:science-default}
    \end{subfigure}
    
    \begin{subfigure}{0.23\textwidth}
        \centering
        \includegraphics[width=\textwidth]{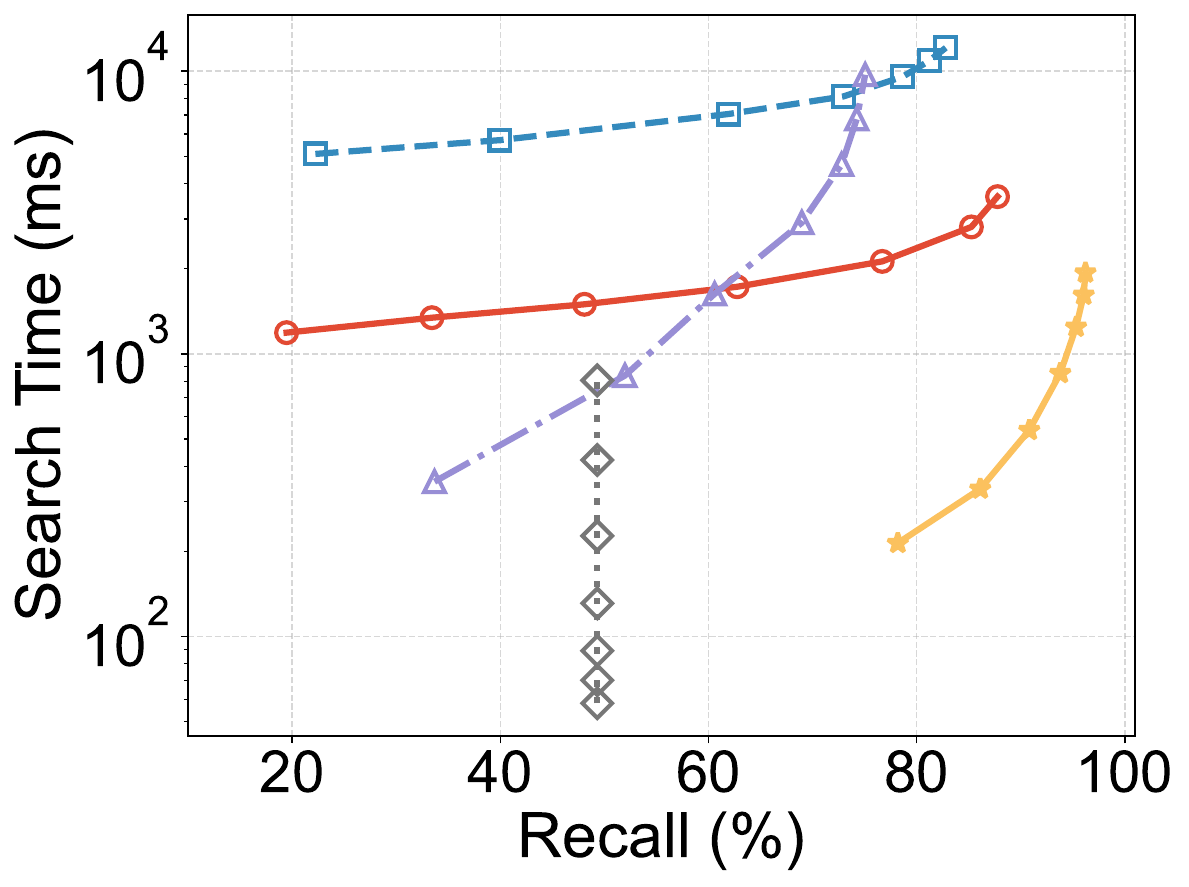}\vspace{-1.0ex}
        \caption{\technology}
        \label{fig:technology-default}
    \end{subfigure}
    \hfill
    \begin{subfigure}{0.23\textwidth}
        \centering
        \includegraphics[width=\textwidth]{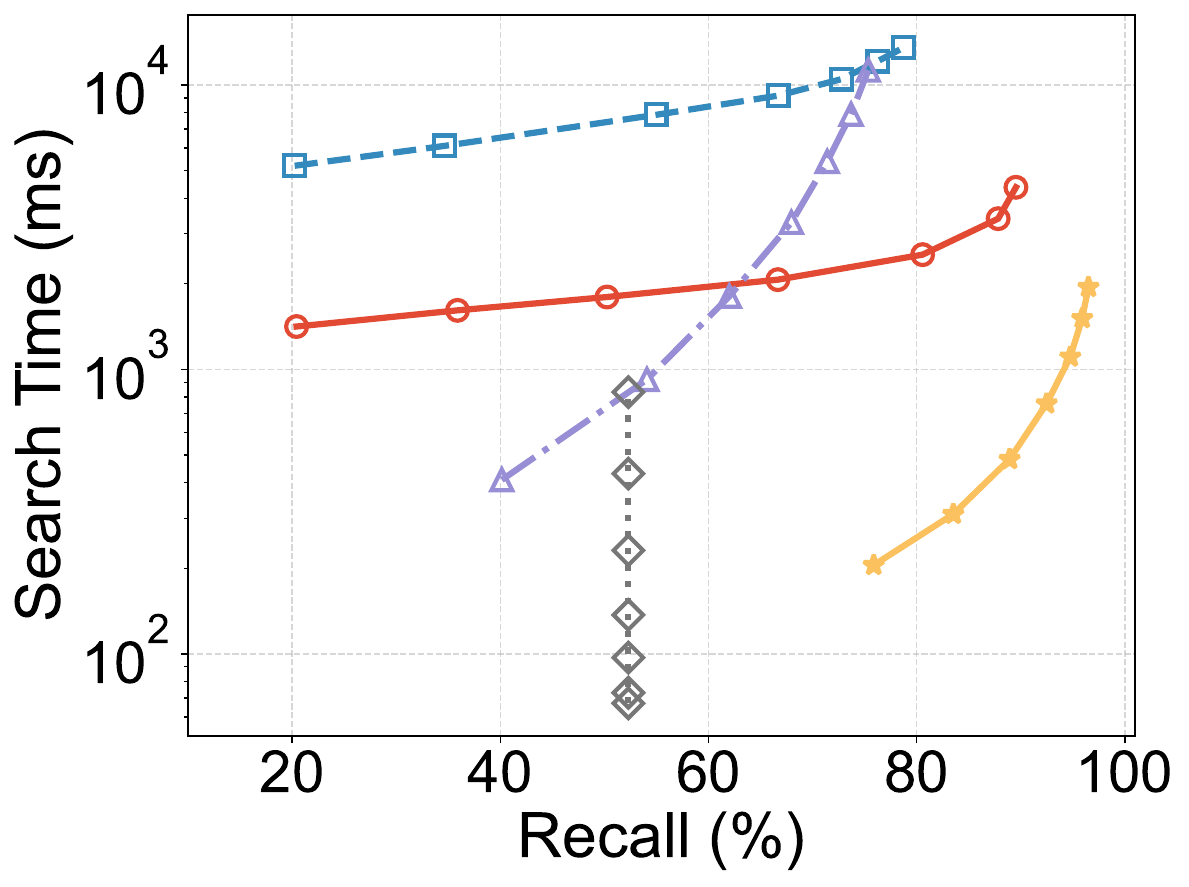}\vspace{-1.0ex}
        \caption{\pooled}
        \label{fig:pooled-default}
    \end{subfigure}
    \hfill
    \begin{subfigure}{0.23\textwidth}
        \centering
        \includegraphics[width=\textwidth]{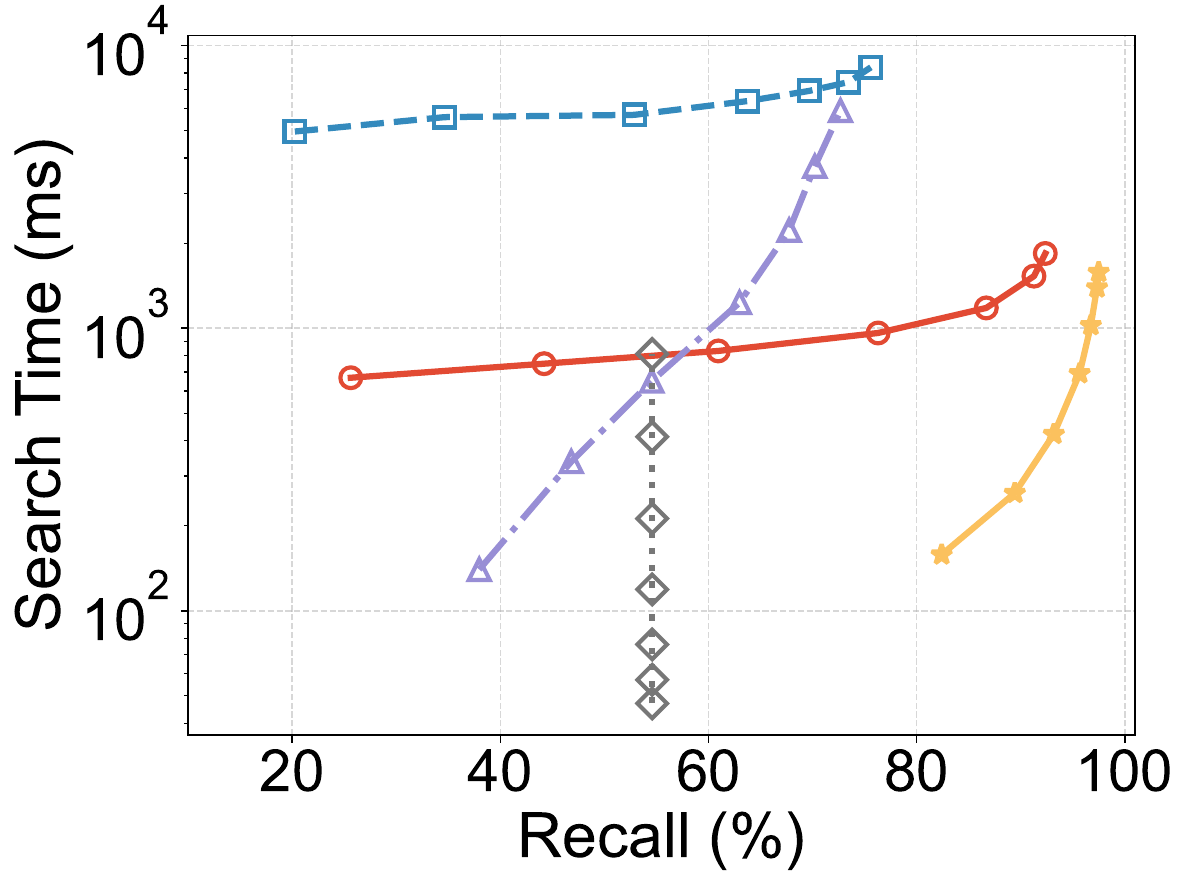}\vspace{-1.0ex}
        \caption{\msmacro}
        \label{fig:msmacro-default}
    \end{subfigure}
    \hfill
    \begin{subfigure}{0.23\textwidth}
        \centering
        \includegraphics[width=0.73\textwidth]{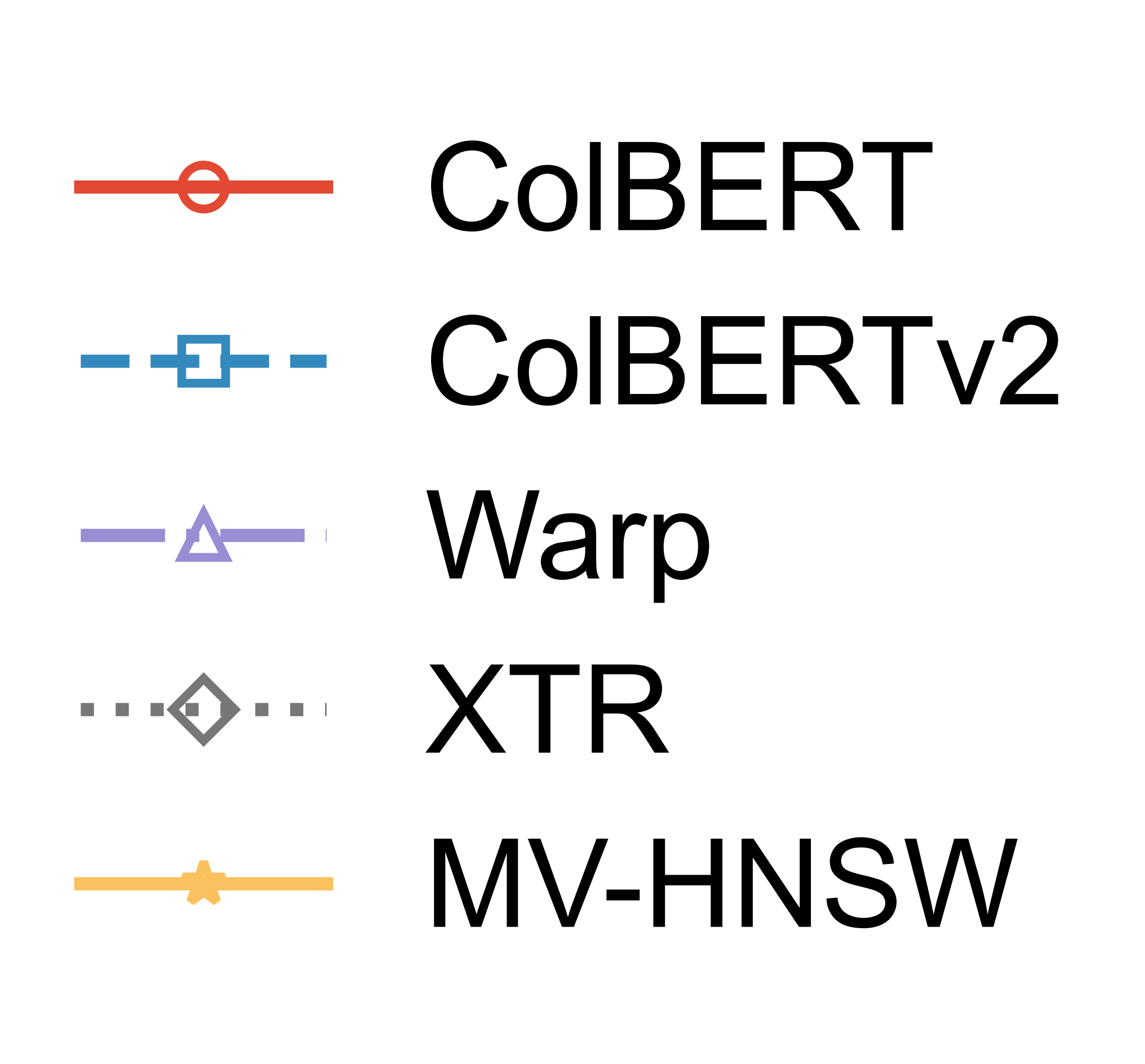}
        \caption*{}
    \end{subfigure}
    \vspace{-2.5ex}
    \caption{Search performance on all seven real-world datasets (with default query parameter $k=128$)}\label{fig:search-performance-default}
    \ybh{\vspace{-3ex}}
\end{figure*}

\begin{figure*}[t]
    \centering
    \begin{subfigure}{0.24\textwidth}
        \centering
        \includegraphics[width=\textwidth]{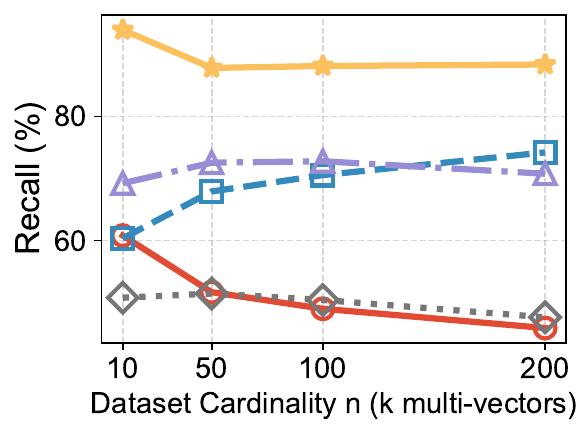}\vspace{-1.0ex}
        \caption{Recall (\pooled{})}
        \label{fig:pooled-varied-datasetsize-recall}
    \end{subfigure}
    \hfill
    \begin{subfigure}{0.24\textwidth}
        \centering
        \includegraphics[width=\textwidth]{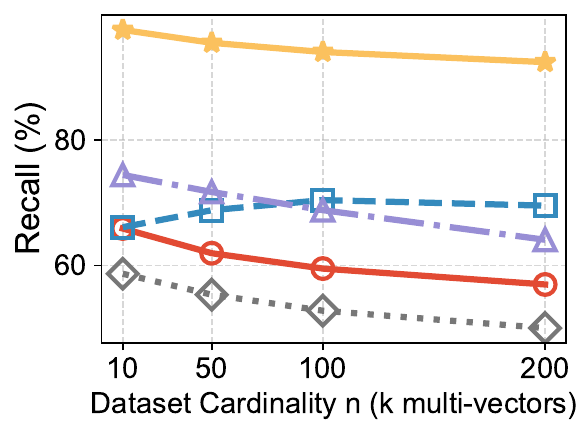}\vspace{-1.0ex}
        \caption{Recall (\msmacro{})}
        \label{fig:msmacro-varied-datasetsize-recall}
    \end{subfigure}
    \hfill
    \begin{subfigure}{0.24\textwidth}
        \centering
        \includegraphics[width=\textwidth]{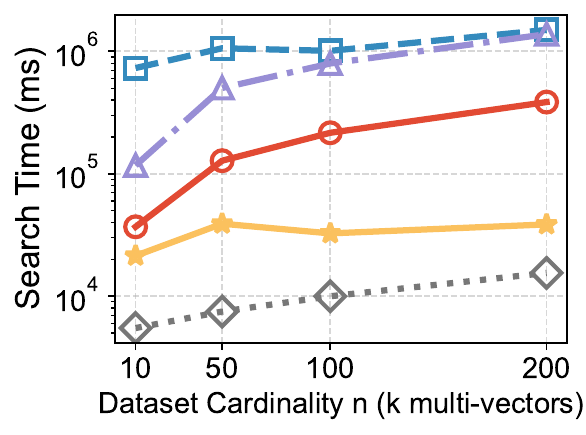}\vspace{-1.0ex}
        \caption{Time (\pooled{})}
        \label{fig:pooled-varied-datasetsize-time}
    \end{subfigure}
    \hfill
    \begin{subfigure}{0.24\textwidth}
        \centering
        \includegraphics[width=\textwidth]{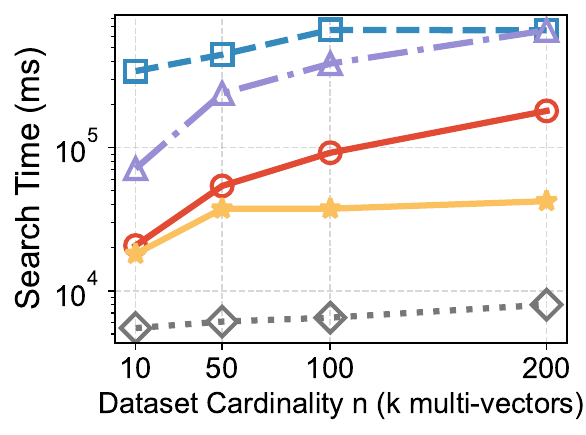}\vspace{-1.0ex}
        \caption{Time (\msmacro{})}
        \label{fig:msmacro-varied-datasetsize-time}
    \end{subfigure}
    \vspace{-2.5ex}
    \caption{Scalability test with varying dataset cardinality $n$}\label{fig:varied-dataset-size}
    \ybh{\vspace{-3ex}}
\end{figure*}

\fakeparagraph{Compared Algorithms} 
We select four representative methods for multi-vector similarity search as the major competitors, including \colbert \cite{khattab2020colbert}, \colbertvii \cite{santhanam2022colbertv2}, \xtr \cite{lee2023rethinking}, and \warp \cite{scheerer2025warp}.
The parameter configurations for each algorithm follow their implementation recommendations.

\begin{itemize} 
    \item \textbf{\colbert} \cite{khattab2020colbert}: This work introduce the multi-vector similarity search problem and the $\MaxSim$ operator. This algorithm retrieves $\alpha \cdot k$ candidate \zeng{token} vector for each query \zeng{token} vector using HNSW \cite{malkov2018efficient}, then re-ranks the corresponding multi-vectors by computing the exact $\operatorname{\USim}$ value.
    Following the original paper, we set $\alpha=0.5$. 
    \zeng{The HNSW parameters are set to $M = 32$ (node degree) and $efConstruction = efSearch = 40$ (candidate list sizes).}

    \item \textbf{\colbertvii} \cite{santhanam2022colbertv2}: This work introduces residual compression that clusters embeddings to reduce storage, along with a centroid-based index to improve  efficiency. We adopt the open-source implementation \cite{colbertviicode} with default parameters.
    
    \item \textbf{\xtr} \cite{lee2023rethinking}: This solution was developed by Google DeepMind and inspired recent state-of-the-art (\eg WARP \cite{scheerer2025warp} and ColXTR \cite{bhat2025xtr}).
    \zeng{Like \colbert{}, \xtr{} retrieves candidate multi-vectors by collecting token vectors similar to each query token vector.
    However, unlike the re-rank procedure of \colbert{}, \xtr{} computes an approximate $\operatorname{\USim}$ value using only the token vectors already retrieved.}
    
    \item \textbf{\warp} \cite{scheerer2025warp}: \warp is a state-of-the-art solution designed for large-scale datasets using the IVF index. 
    We set $n_{probe}=32$ (clusters probed), $t'=40000$ (cumulative cluster size threshold), and $b=4$ (residual compression rate).
\end{itemize}

Beyond these algorithms, we also compare our solution with Qdrant \cite{qdrant}, a vector database that provides a ColBERT-based solution for multi-vector similarity search. Although Qdrant incorporates system-level optimizations, experimental results show that our {\IndexName} achieves superior search performance (see our full paper \cite{mvhnsw-2026} for details).

\fakeparagraph{Metrics} 
We assess the performance of all methods based on two criteria: \textit{search performance} and \textit{index construction cost}. 
For \textit{search performance}, we use \textbf{recall} (defined in \equref{equ:recall}) to measure the accuracy of the retrieved results, and \textbf{search time} (average time per query) to measure time efficiency.
For \textit{index construction cost}, we measure \textbf{construction time} and \textbf{index size}. 

\fakeparagraph{Implementation and Environment} 
All methods, including our proposed approach and baselines, are implemented in C++ and Python, with C++ components compiled using g++. 
Experiments are conducted on a server equipped with 32 AMD Ryzen 9 5950X 3.40GHz Processor with 128 GB RAM, running Ubuntu 24.04.2.
For index construction, our algorithm leverages multi-threading, while all other methods and all search experiments are conducted in a single-threaded environment. 
Our source code and datasets are publicly  available on GitHub~\cite{code}.



\vspace{-1ex}
\subsection{Overall Search Performance}
We evaluate search performance by measuring recall and search time and illustrating the results using \textit{recall-time curves}. 
A \textit{time-recall curve} visualizes the trade-off between accuracy and efficiency, with recall on the x-axis and search time (in milliseconds) on the y-axis. 
This representation highlights how much time an algorithm requires to achieve a given recall level. 
Consequently, a curve that is shifted towards the \textit{bottom-right corner indicates superior performance}, as it achieves higher recall with lower latency.

\figref{fig:search-performance-default} shows the recall-time curves on all the datasets under the default parameter settings: the number of required nearest neighbors $k = 128$, the dimensionality of each \zeng{token} vector $d = 128$, and the number of \zeng{token} vectors in each multi-vector $c = 32$.
As the results show, our method \MethodName{} consistently outperforms the baselines across all datasets.


\textbf{Recall.} 
Our \ybh{\MethodName} achieves substantially higher recall than the compared baselines. 
On all seven datasets, it attains over 90\% recall with minimal latency. 
In contrast, the baselines generally fall into two categories: \textit{high-latency} and \textit{ultra-fast} methods.
High-latency baselines, such as \colbertvii, prioritize recall but incur significant computational overhead.
When fixing search latency at \zeng{1.2 to 3.9 seconds}, their recall saturates at around 80\%, and even at this level, their latency remains 2.4--25.2$\times$ higher than that of our method at 90\% recall. 
Ultra-fast baselines (represented by \xtr) sacrifice recall (hovering around 50\%) for extremely low latency (10 to 100 ms) but fail to scale to higher recall tiers.

\textbf{Search Latency.} 
Our method demonstrates superior efficiency when targeting high recall. 
While baseline methods exhibit explosive latency growth as recall increases, our approach maintains consistently low latency.
For example, to reach a target recall of 90\% across all datasets, our method achieves a speedup of 5.6--14.0$\times$ over the fastest competing baseline.

\begin{figure*}[t]
\vspace{-0.5ex}
\includegraphics[width=0.65\textwidth]{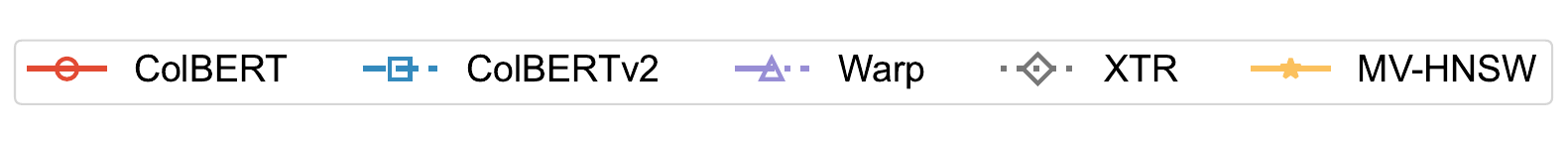}
\vspace{-1.5ex}
    
    \centering
    \begin{subfigure}{0.24\textwidth}
        \centering
        \includegraphics[width=\textwidth]{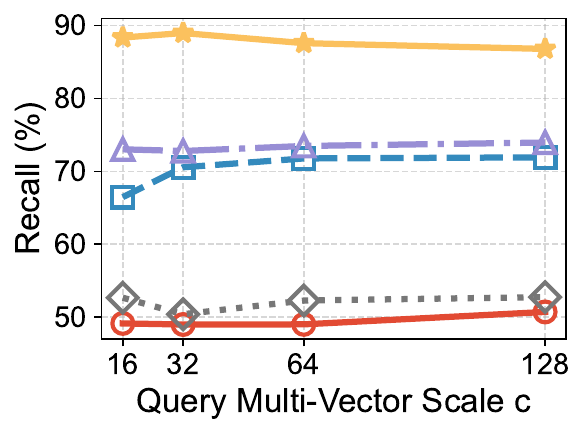}\vspace{-1.0ex}
        \caption{Recall (\pooled{})}
        \label{fig:pooled-varied-q-recall}
    \end{subfigure}
    \begin{subfigure}{0.24\textwidth}
        \centering
        \includegraphics[width=\textwidth]{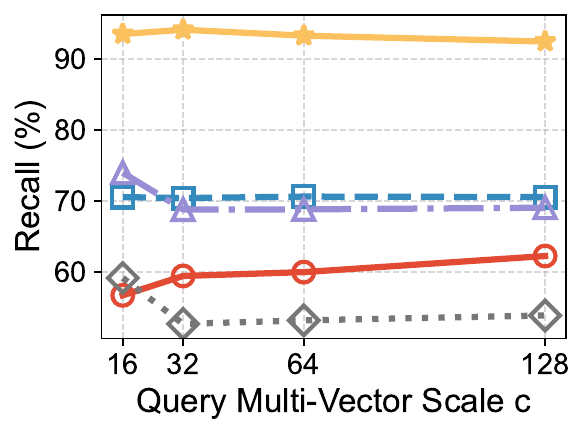}\vspace{-1.0ex}
        \caption{Recall (\msmacro{})}
        \label{fig:msmacro-varied-q-recall}
    \end{subfigure}
    \begin{subfigure}{0.24\textwidth}
        \centering
        \includegraphics[width=\textwidth]{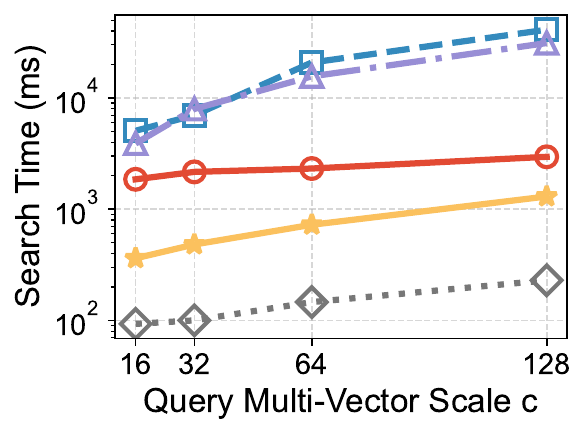}\vspace{-1.0ex}
        \caption{Time (\pooled{})}
        \label{fig:pooled-varied-q-time}
    \end{subfigure}
    \begin{subfigure}{0.24\textwidth}
        \centering
        \includegraphics[width=\textwidth]{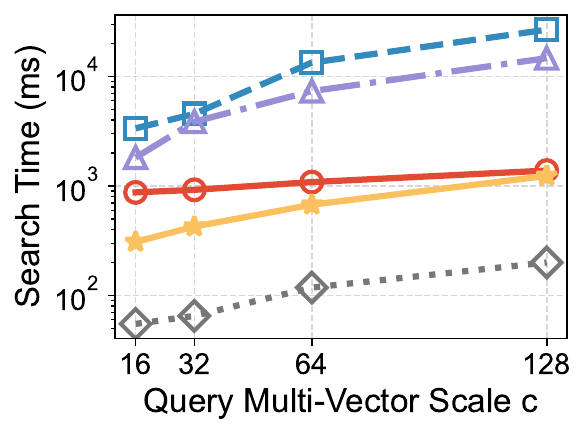}\vspace{-1.0ex}
        \caption{Time (\msmacro{})}
        \label{fig:msmacro-varied-q-time}
    \end{subfigure}
    \vspace{-2.5ex}
    \caption{Scalability test with varying the number $c$ of \zeng{token} vectors per query multi-vector}\label{fig:varied-qlen}
    \vspace{-3ex}
\end{figure*}

\begin{figure*}[t]
\begin{subfigure}{0.24\textwidth}
        \centering
        \includegraphics[width=\textwidth]{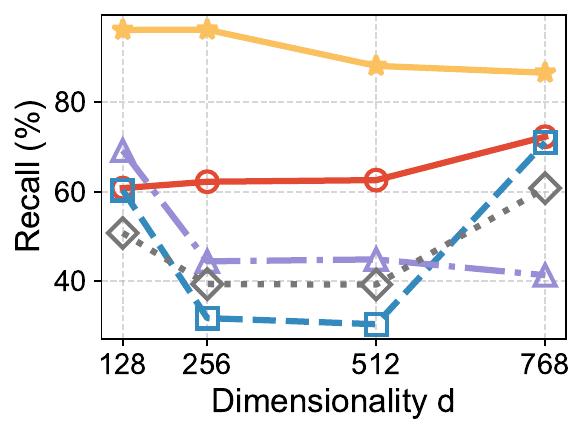}\vspace{-1.0ex}
        \caption{Recall (\pooled{})}
        \label{fig:pooled-varied-dim-recall}
    \end{subfigure}
    \begin{subfigure}{0.24\textwidth}
        \centering
        \includegraphics[width=\textwidth]{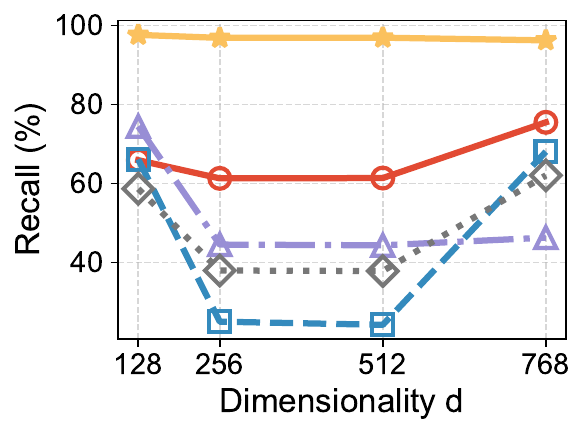}\vspace{-1.0ex}
        \caption{Recall (\msmacro)}
        \label{fig:msmacro-varied-dim-recall}
    \end{subfigure}
    \centering
    \begin{subfigure}{0.24\textwidth}
        \centering
        \includegraphics[width=\textwidth]{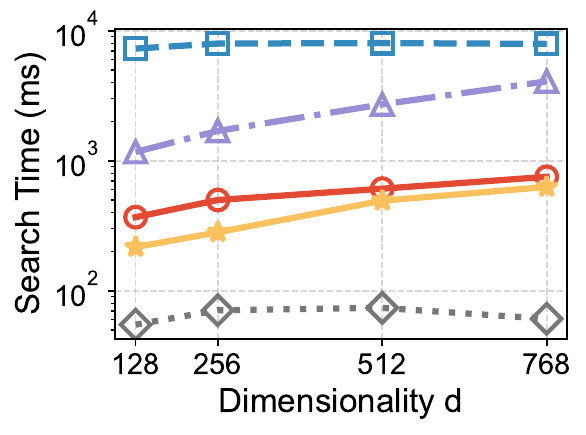}\vspace{-1.0ex}
        \caption{Time (\pooled{})}
        \label{fig:pooled-varied-dim-time}
    \end{subfigure}
    \begin{subfigure}{0.24\textwidth}
        \centering
        \includegraphics[width=\textwidth]{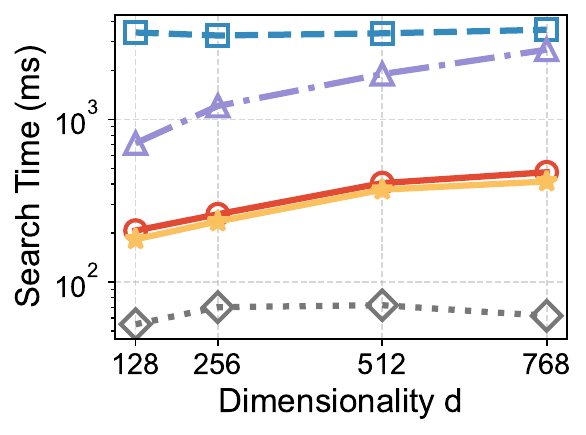}\vspace{-1.0ex}
        \caption{Time (\msmacro{})}
        \label{fig:msmacro-varied-dim-time}
    \end{subfigure}
    \vspace{-2.5ex}
    \caption{Scalability test with varying the dimensionality $d$ of \zeng{token} vectors}\label{fig:varied-dim}
    \vspace{-3ex}
\end{figure*}

\begin{figure*}[t]
    \centering
    \begin{subfigure}{0.24\textwidth}
        \centering
        \includegraphics[width=\textwidth]{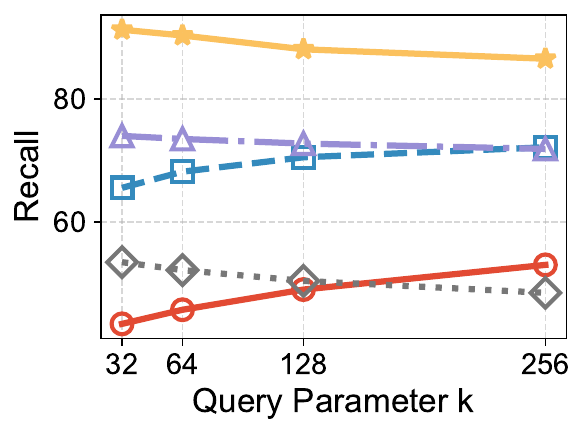}\vspace{-1.0ex}
        \caption{Recall (\pooled{})}
        \label{fig:pooled-varied-k-recall}
    \end{subfigure}
    \begin{subfigure}{0.24\textwidth}
        \centering
        \includegraphics[width=\textwidth]{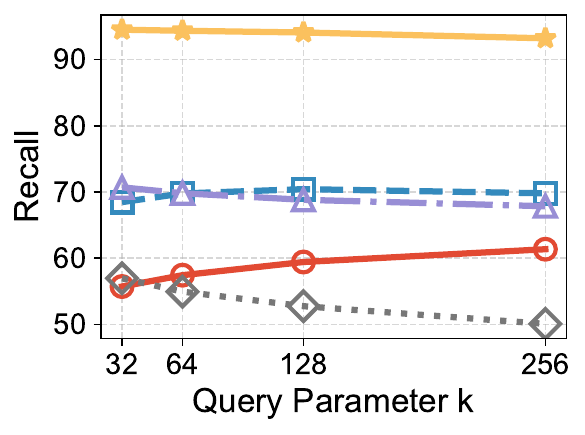}\vspace{-1.0ex}
        \caption{Recall (\msmacro{})}
        \label{fig:msmacro-varied-k-recall}
    \end{subfigure}
    \begin{subfigure}{0.24\textwidth}
        \centering
        \includegraphics[width=\textwidth]{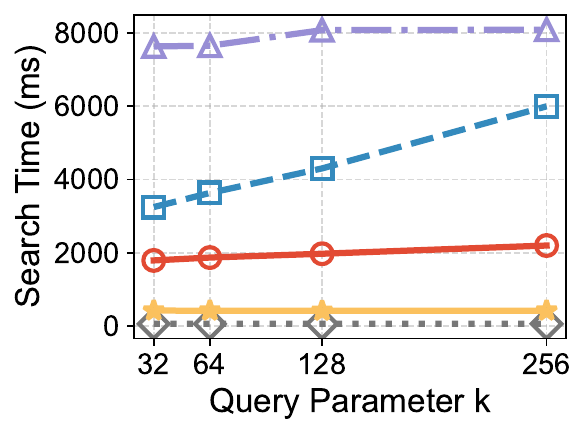}\vspace{-1.0ex}
        \caption{Time (\pooled)}
        \label{fig:pooled-varied-k-time}
    \end{subfigure}
    \begin{subfigure}{0.24\textwidth}
        \centering
        \includegraphics[width=\textwidth]{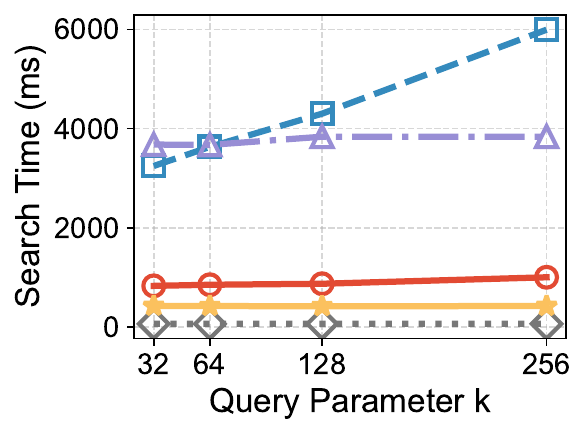}\vspace{-1.0ex}
        \caption{Time (\msmacro{})}
        \label{fig:msmacro-varied-k-time}
    \end{subfigure}
    \vspace{-2.5ex}
    \caption{Impact of the number $k$ of nearest neighbors}\label{fig:varied-query-k}
    \vspace{-3ex}
\end{figure*}

\subsection{Scalability Test}

We evaluate scalability on on \pooled{} and \msmacro{} datasets by varying three key parameters: the total number $n$ of multi-vectors, the number $c$ of \zeng{token} vectors per query multi-vector, and the dimensionality $d$ of each \zeng{token} vector. 
Due to page limitations, please refer to our full paper \cite{mvhnsw-2026} for results on the other datasets.

\subsubsection{Varying Dataset Cardinality.} 
The total number of vectors in the dataset (\ie cardinality) is the primary metric for evaluating data scalability. 
To assess this, we conduct scalability tests on dataset subsets of varying sizes, with $n \in \{1 \times 10^4, 5 \times 10^4, 1 \times 10^5, 5 \times 10^5\}$, and compare our method against the baselines.

\textbf{Latency Scaling.} 
Across all datasets, our method exhibits slow and stable growth in search latency as the dataset size increases. 
Specifically, the latency growth trend remains logarithmic in log-scale coordinates, which closely aligns with the time complexity analysis presented in \secref{sec:search-complexity}. 
Overall, our method achieves a significant latency advantage over all baselines, with the sole exception of \xtr{}.
However, the low recall of \xtr{} undermines its practicality, underscoring the importance of the recall-latency balance achieved by our approach.

\textbf{Recall Stability.} 
Our method maintains robust recall performance as the dataset scale expands. 
The overall recall trend of our \IndexName remains consistently high, staying at approximately 90\% across all dataset sizes and vastly outperforming all compared algorithms.
This demonstrates the superior scalability of our {\IndexName} in preserving search accuracy over large-scale datasets.

\subsubsection{Varying the Number of \zeng{Token} Vectors Per Query Multi-Vector.}
In practice, the number $c$ of \zeng{token} vectors per query multi-vector directly impacts search latency. Accordingly, we evaluate search performance by varying $c \in \{16, 32, 64, 128\}$ on the datasets.
\figref{fig:varied-qlen} illustrates the search time and recall results, respectively.

\textbf{Search Latency.} 
Our method exhibits efficient sub-linear time complexity as the number of \zeng{token} vectors per query increases. 
It consistently outperforms all compared algorithms except \xtr{}.
As the parameter $c$ doubles, the growth rate of our search time remains significantly lower than that of the baselines, demonstrating a sub-linear trend that aligns with our time complexity analysis.
When $c$ reaches the maximum $128$, our method achieves a speedup up to 2.3$\times$ over \colbert{} and a speedup of 11.9--31.6$\times$ over \warp{} and \colbertvii{}. 
However, as discussed below, \xtr{}'s extreme speed comes at the cost of severely limited recall.

\textbf{Recall Stability.} 
Our method maintains exceptional recall stability regardless of the number $c$ of \zeng{token} vectors per query, consistently achieving approximately 90\% recall across all $c$ values.
In contrast, the maximum recall of other algorithms does not exceed 80\%.
Notably, while \xtr{} exhibits high efficiency, its recall remains below 60\% across the tests and drops below 50\% at certain $c$ values.

\vspace{-2ex}
\subsubsection{Search Performance on Varied Dimensionality.} 
Dimensionality $d$ is another natural metric for scalability evaluation, and we evaluate the search performance by varying $d \in \{128, 256, 512, 768\}$ over a sample of each dataset. 
\figref{fig:varied-dim} presents the corresponding recall and search time.

\textbf{Recall.} As $d$ varies, our algorithm consistently maintains recall around 95\% and never drops below 90\%, demonstrating its advantage on search accuracy. 
In contrast, the recall of the baselines fluctuates significantly between 20\% and 80\%, showing high sensitivity to dimensionality changes.

\textbf{Search Time.} Across all tested dimensions, our algorithm outperforms all baselines in terms of search time, with the sole exception of \xtr{}. Although \xtr{} achieves the optimal efficiency, its recall fluctuates drastically with varying $d$ and remains below 60\% on most datasets.

\vspace{-2ex}
\subsection{Impact of Query Configuration}

This experiment evaluates the impact of the query parameter \(k\) and a different similarity metric on the \pooled{} and \msmacro{} datasets. 
Due to page limitations, please refer to \cite{mvhnsw-2026} for other results.

\subsubsection{Impact of Query Parameter $k$.}
\(k\) controls the number of nearest neighbors returned. To evaluate sensitivity under varying \(k\), we conduct experiments with \(k \in \{32, 64, 128, 256\}\) in \figref{fig:varied-query-k}. 

\textbf{Search Time.} 
The search time of our method is stable across different values of $k$. 
Unlike \colbertvii{}, which exhibits (nearly) linear search time growth as $k$ increases, our search time remains nearly flat and constant even as $k$ doubles. 
Moreover, our method is significantly faster than all high-recall baselines, such as \ybh{\colbertvii{} and \warp{}}.
For example, it achieves an 6.9--19.7$\times$ speedup over \colbertvii{}. 
While \xtr{} achieves lower latency, it sacrifices too much accuracy to be practically viable, particularly as its recall degrades further with larger $k$.

\textbf{Recall.} 
Our method strictly dominates all baselines in recall as $k$ scales. 
It consistently maintains a recall of approximately 90\%, surpassing the 80\% ceiling observed across all compared algorithms. 
Furthermore, our method exhibits strong stability: as $k$ grows exponentially, the recall curve remains smooth and shows minimal degradation. This robustness is attributed to the effectiveness of our underlying graph-based index structure and optimized search algorithm, which ensure that the candidate set captures the majority of ground truth objects efficiently.

\subsubsection{Beyond MaxSim Multi-vector Similarity.}

We conducted experiments on another multi-vector similarity function, aggregate $\gamma$NN \cite{lee2023rethinking}. Since other algorithms do not natively support aggregated \(\gamma\)NN, we only compare against \xtr{}.

As shown in \tabref{tab:gamma-performance}, our algorithm achieves a recall of over 94.9\% across both datasets, substantially outperforming \xtr{}, which attains only 50.7\%--58.9\%. This represents a relative improvement of 60\%--87\% in recall. In terms of query latency, our algorithm is approximately 3--5\(\times\) slower than \xtr{}. However, this latency remains well within an acceptable range for practical applications, especially given the significant recall gains.


\begin{table}[t]
\caption{Experiment on Aggregate $\gamma$NN (recall: \%, time: ms)}\label{tab:gamma-performance}
\vspace{-2.5ex}
\centering
\footnotesize 
\setlength{\tabcolsep}{3.5pt} 
\begin{tabular}{l cccccccc}
\toprule
\multirow{3}{*}{\begin{tabular}[c]{@{}l@{}}Compared\\ Algorithms\end{tabular}} & \multicolumn{4}{c}{Macro} & \multicolumn{4}{c}{Pooled} \\
\cmidrule(lr){2-5} \cmidrule(lr){6-9}
& \multicolumn{2}{c}{$\gamma=2$} & \multicolumn{2}{c}{$\gamma=8$} & \multicolumn{2}{c}{$\gamma=2$} & \multicolumn{2}{c}{$\gamma=8$} \\
\cmidrule(lr){2-3} \cmidrule(lr){4-5} \cmidrule(lr){6-7} \cmidrule(lr){8-9}
& Recall & Time & Recall & Time & Recall & Time & Recall & Time \\
\midrule
Linear Scan & - & 862.6 & - & 1104.1 & - & 1589.0 & - & 1912.6 \\
\midrule
XTR & 58.9 & 51.9 & 54.3 & 51.7 & 50.7 & 54.5 & 54.3 & 51.7 \\
MV-HNSW & 96.0 & 173.7 & 96.0 & 216.1 & 94.9 & 228.8 & 95.3 & 260.3 \\
\bottomrule
\end{tabular}
\end{table}


\begin{table}[t]
\caption{Index construction (time: h, size: GB)}\label{tab:index-construction}
\vspace{-2.5ex}
\centering
\setlength{\tabcolsep}{3.5pt} 
\begin{tabular}{l cccccccc}
\toprule
\multirow{2}{*}{\begin{tabular}[c]{@{}l@{}}Compared\\ Algorithms\end{tabular}} & \multicolumn{2}{c}{\writing} & \multicolumn{2}{c}{\science} & \multicolumn{2}{c}{\pooled} & \multicolumn{2}{c}{\msmacro} \\
\cmidrule(lr){2-3} \cmidrule(lr){4-5} \cmidrule(lr){6-7} \cmidrule(lr){8-9}
& Time & Size & Time & Size & Time & Size & Time & Size \\
\midrule
\colbert{}   & 1.85 & 3.58 & 2.18 & 4.11 & 2.04 & 3.93 & 0.92 & 1.76 \\
\colbertvii{} & 0.37 & 0.50 & 0.41 & 0.22 & 1.29 & 0.55 & 0.58 & 0.25 \\
\xtr{}       & 6.62 & 3.60 & 6.09 & 4.13 & 6.58 & 3.95 & 7.74 & 1.76 \\
\warp{}      & 0.69 & 1.06 & 0.81 & 1.21 & 0.78 & 0.95 & 0.33 & 0.52 \\
\IndexName{} & 21.29 & 7.65 & 21.06 & 6.22 & 18.09 & 6.90 & 13.81 & 10.73 \\
\bottomrule
\end{tabular}
\end{table}

\vspace{-1ex}
\subsection{Experiment on Index Construction}




This experiment evaluates the index construction time and index size on on \writing{}, \science{}, \pooled{}, and \msmacro{} datasets. 
Please refer to \cite{mvhnsw-2026} for the complete results.
Overall, while our method incurs higher costs, the overhead remains manageable and acceptable for real-world datasets. 

\textbf{Index Size.} As shown in \tabref{tab:index-construction}, \colbertvii{} and \warp{}, exhibit the most compact index sizes. In contrast, algorithms based on token-vector HNSW (\colbert{} and \xtr{}) incur an index size roughly \(3.4\times\) larger than that of \warp{}. Our {\IndexName} yields a slightly larger index size than \colbert{} and \xtr{}, but remains well within an acceptable range.

\textbf{Construction Time.} Among the baselines, \colbert{}, \colbertvii{}, and \warp{} demonstrate exceptionally high efficiency in index construction. The average construction time for \xtr{} is approximately \(9.5\times\) that of \warp{} and \(3.8\times\) that of \colbert{}. To achieve better online search performance, our proposed algorithm incurs a longer index construction time than these baselines. Notably, our ablation study shows that this overhead can be significantly reduced by using more threads during construction.




\vspace{-1ex}
\subsection{Ablation Study}

We conduct ablation studies to evaluate the effectiveness of the proposed optimization methods: \textit{efficient {\USim} computation} (\secref{subsec:EUC}) and \textit{augmented search layer} (\secref{sec:SearchAugumentation}).

\subsubsection{Optimization \#1: Accelerated {\USim} Computation.} 
This optimization accelerates the computation of the {\USim} operator, thereby improving both index construction time and search latency. 
In this ablation study, we focus on evaluating its impact on index construction, as the acceleration effect is identical for both phases: the optimized {\USim} computation is invoked in the same manner during indexing and searching. 

\figref{fig:ablation-euc} illustrates the index construction performance on the Pooled, Science, and Writing datasets, comparing results with and without optimization \#1. 
To isolate the overhead introduced by multi-vectors, we include an extra baseline (named ``\zeng{token}-only HNSW'') that builds a standard HNSW index over \textit{all \zeng{token} vectors} in each dataset. 
This baseline serves as a reference point for understanding the computational cost of our multi-vector index.

Applying optimization \#1 yields a substantial reduction in index construction time. 
For constructing a 128-dimensional multi-vector indexs with 10,000 data objects on these datasets, the construction time of our {\IndexName} is reduced to 3.1--3.3$\times$ that of \zeng{token}-only HNSW. 
Without optimization \#1, this overhead increases to 10.9--13.6$\times$. 
These results confirm that optimization \#1 effectively mitigates the computational complexity introduced by the multi-vector structure in the graph index.


\begin{figure}[h]
    \vspace{-0.5ex}
    \centering
    \includegraphics[width=0.95\columnwidth]{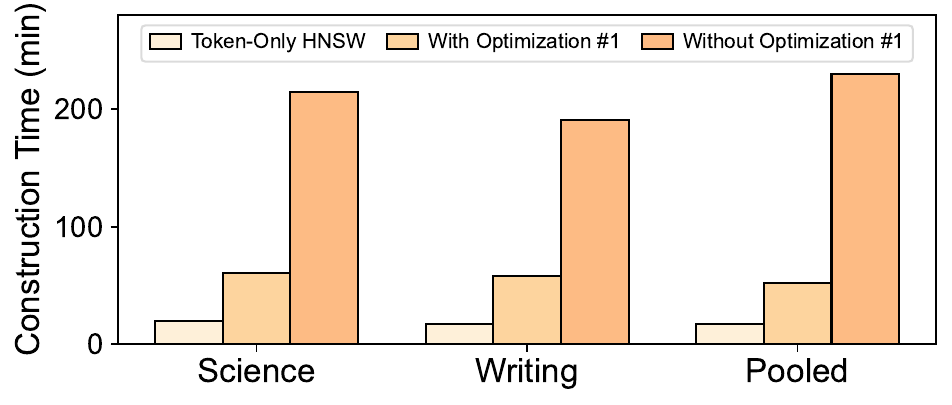}
    \vspace{-2.5ex}
    \caption{Ablation study on Optimization \#1}\label{fig:ablation-euc}
    \vspace{-2ex}
\end{figure}

\subsubsection{Optimization \#2: Augmented Search Layer.} 
This optimization enhances search accuracy by introducing an auxiliary index and a dynamic candidate expansion strategy. 
We compare the standard HNSW search algorithm \cite{malkov2018efficient} with our augmented search strategy on the same {\IndexName} structure constructed over the \lifestyle{}, \writing{}, \technology{}, and \msmacro{} datasets.

\figref{fig:ablation-augmentedsearch} presents the recall results with and without optimization \#2 across all datasets. 
The results demonstrate that our augmented search algorithm consistently improves recall compared to the standard search method.
This is because our algorithm explores candidates that are topologically distant or disconnected from the traversed nodes, whereas the baseline, strictly following connected edges, inevitably becomes trapped in local optima.

\begin{figure}[h]
    \vspace{-2ex}
    \centering
    \includegraphics[width=0.85\columnwidth]{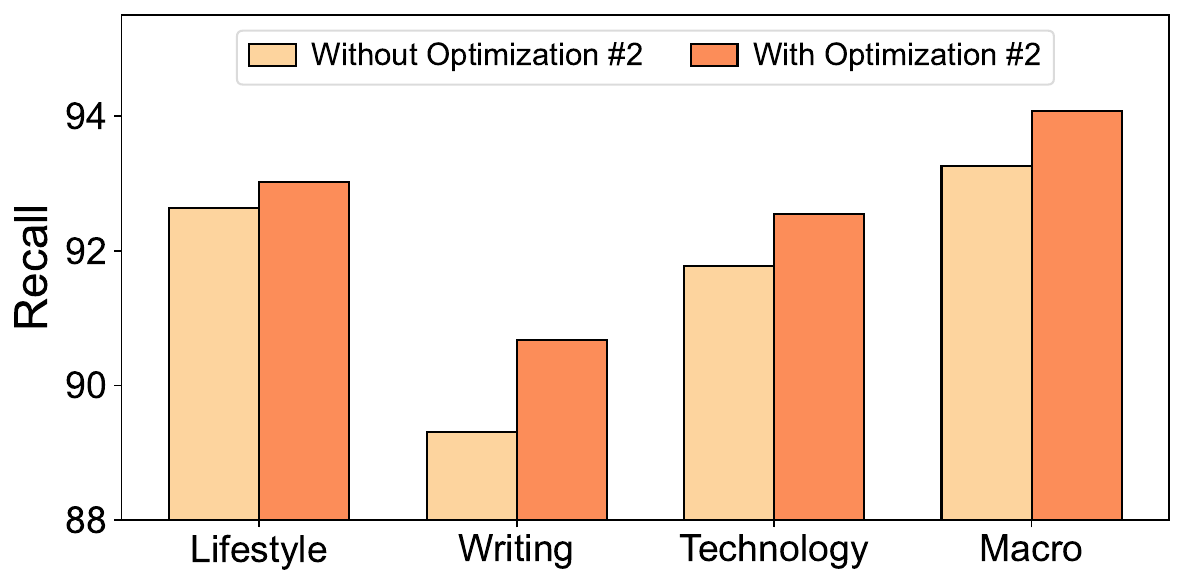}
    \vspace{-3.5ex}
    \caption{Ablation study on Optimization \#2}\label{fig:ablation-augmentedsearch}
    \vspace{-2.5ex}
\end{figure}

\vspace{-1ex}
\subsection{Summary of Major Findings}
Our experimental study yields the following key findings:
\begin{itemize}[noitemsep]
\item \textbf{Balanced Search Performance}. 
Our \IndexName{} achieves the best trade-off between recall and efficiency. 
It usually maintains over 90\% recall with competitive search time and can be 5.6--14.0$\times$ faster than the baselines at the same recall level.
As dataset scale and query parameters vary, our method maintains stable recall with only modest growth in the search latency.


\item \textbf{Acceptable Offline Trade-offs}. 
The superior online search performance comes at the increased cost of offline indexing. Although our index construction time and size exceed those of some baselines, they remain practical for real-world deployments, thanks to our accelerated $\operatorname{\USim}$ computation optimization, which achieves up to 10.9--13.6$\times$ speedup.

\item \textbf{Baseline Comparison}. 
Among the baselines, \colbert{} strikes the most balanced trade-off between recall and search time. \xtr{} prioritizes efficiency at the expense of recall, while \colbertvii{} and \warp{} sacrifice efficiency to achieve moderate recall levels. Overall, none of these baselines achieves consistently satisfactory search performance.
\end{itemize}

\vspace{-1.5ex}
\section{\MakeUppercase{Related Work}}	\label{sec:related-work}

We review related work in the following two categories: \textit{multi-vector similarity search} and \textit{approximate nearest neighbor search}.

\fakeparagraph{Multi-Vector Similarity Search} 
Prior studies on this problem adopt a filter-and-refine framework.
The filter phase is used to prune the entire dataset to a small set of candidate multi-vectors, avoiding the expensive {\USim} computation over all multi-vectors. 
In the refine phase, the exact {\USim} score is computed between the query and all candidates generated to determine the final results.

Based on the filtering strategy, existing work can be broadly classified into three categories: \textit{bound estimation}, \textit{clustering-based}, and \textit{lexical-based} methods.
\textit{Bound estimation methods}~\cite{khattab2020colbert, lee2023rethinking, li2023slim} estimate lower bounds of {\USim} using single-vector distances between query and data objects. 
\textit{Clustering-based methods}~\cite{gao2021coil, li2023citadel} focus on \zeng{token-vector-level} matching, pruning objects based on single-vector distances when their original tokens exactly match or are semantically similar.
\textit{Lexical-based methods} \cite{santhanam2022colbertv2, santhanam2022plaid, scheerer2025warp} rely on clustering and quantization techniques. 
These methods exploit semantic locality by clustering \zeng{token} vectors within each multi-vector and quantizing their residuals to reduce storage cost. 
During search, \zeng{multi-vectors} are pruned if they contain no \zeng{token} vectors belonging to centroids close to the query's \zeng{token} vectors.
We select the state-of-the-art method from each category as our primary competitors: \colbert{} \cite{khattab2020colbert}, \colbertvii{} \cite{santhanam2022colbertv2}, \xtr{} \cite{lee2023rethinking}, and \warp{} \cite{scheerer2025warp}.

Existing methods prune multi-vectors by leveraging \zeng{token} vectors but often overlook the rich inter-multi-vector information, leading to the inadvertent pruning of relevant candidates and thus limiting search performance. 
By contrast, our solution organizes data objects at the multi-vector level and navigates the search process accordingly, thereby enhancing both recall and efficiency.

\fakeparagraph{Approximate Nearest Neighbor Search} 
Approximate nearest neighbor search (ANNS) indexes are extensively leveraged in existing multi-vector similarity search algorithms. 
Josef \etal \cite{sivic2003video} pioneered the use of k-means clustering for high-dimensional ANN search and introduced the Inverted File (IVF) index to record vector IDs belonging to each cluster centroid. 
Jegou \etal \cite{jegou2010product} advanced this solution by representing vectors as residuals to their cluster centroids and applying Product Quantization (PQ) to compress these residuals.
Wang \etal \cite{wang2020deltapq} later observed that residuals are more concentrated within the IVF framework, enabling further compression for improved efficiency. 
In multi-vector similarity search, clustering-based methods \cite{khattab2020colbert, lee2023rethinking, li2023slim} adopt a similar filtering paradigm. Additionally, some lexical-based methods, such as COIL \cite{gao2021coil}, employ IVF to manage vector IDs associated with each token type.

Graph-based indexes offer a superior accuracy-efficiency trade-off among ANN methods \cite{wang2021comprehensive}, with representative solutions like HNSW \cite{malkov2018efficient}, LSH-APG~\cite{lshapg}, and SHG~\cite{DBLP:journals/pvldb/GongZC25}.
Among these, HNSW is widely regarded as the state-of-the-art \cite{pan2024survey} and is widely used in bound estimation methods for multi-vector similarity search to manage \zeng{token} vectors and derive heuristic information. 
However, despite the success of graph-based indexes in the ANN setting, no existing work applies graph indexing directly to multi-vector data. 




\vspace{-2ex}
\section{\MakeUppercase{Conclusion}}\label{sec:CONCLUSION} 


Fine-grained semantic retrieval in RAG and LLM systems increasingly demands accurate and scalable multi-vector similarity search. Existing approaches predominantly follow a filter-and-refine framework: they first apply heuristics based on single-vector similarity to filter candidates, then refine the results. 
However, this heuristic filter stage ignores the rich information embedded in multi-vector data, leading to solutions that are either inaccurate or inefficient.
To address these limitations, we formalize the Unified Multi-Vector Similarity Search problem.
We propose {\IndexName}, a hierarchical graph index tailored for multi-vector data using a novel edge-weight function that rigorously satisfies three core properties: symmetry, cardinality robustness, and query consistency.
We further develop optimized algorithms for index construction and query processing, and provide a thorough complexity analysis.
Experimental results on seven real-world datasets demonstrate that our solution achieves a superior balance between accuracy and efficiency compared to state-of-the-art methods. 
{\IndexName} accelerates search latency by up to \zeng{5.6--14.0$\times$} while stably maintaining over 90\% recall.


\balance
\bibliographystyle{ACM-Reference-Format}
\bibliography{references}

@inproceedings{khattab2020colbert,
  title={Colbert: Efficient and effective passage search via contextualized late interaction over bert},
  author={Khattab, Omar and Zaharia, Matei},
  booktitle={Proceedings of the 43rd International ACM SIGIR conference on research and development in Information Retrieval},
  pages={39--48},
  year={2020}
}

@inproceedings{santhanam2022colbertv2,
  title={ColBERTv2: Effective and Efficient Retrieval via Lightweight Late Interaction},
  author={Santhanam, Keshav and Khattab, Omar and Saad-Falcon, Jon and Potts, Christopher and Zaharia, Matei},
  booktitle={Proceedings of the 2022 Conference of the North American Chapter of the Association for Computational Linguistics: Human Language Technologies},
  pages={3715--3734},
  year={2022}
}

@inproceedings{santhanam2022plaid,
  title={PLAID: an efficient engine for late interaction retrieval},
  author={Santhanam, Keshav and Khattab, Omar and Potts, Christopher and Zaharia, Matei},
  booktitle={Proceedings of the 31st ACM International Conference on Information \& Knowledge Management},
  pages={1747--1756},
  year={2022}
}

@inproceedings{louis2025colbert,
  title={ColBERT-XM: A Modular Multi-Vector Representation Model for Zero-Shot Multilingual Information Retrieval},
  author={Louis, Antoine and Saxena, Vageesh Kumar and van Dijck, Gijs and Spanakis, Gerasimos},
  booktitle={Proceedings of the 31st International Conference on Computational Linguistics},
  pages={4370--4383},
  year={2025}
}

@inproceedings{faysse2025colpali,
  title={ColPali: Efficient Document Retrieval with Vision Language Models},
  author={Faysse, Manuel and Sibille, Hugues and Wu, Tony and Omrani, Bilel and Viaud, Gautier and Hudelot, C{\'e}line and Colombo, Pierre},
  booktitle={ICLR},
  year={2025}
}

@inproceedings{park2025scv,
  title={SCV: Light and Effective Multi-Vector Retrieval with Sequence Compressive Vectors},
  author={Park, Cheoneum and Jeong, Seohyeong and Kim, Minsang and Lim, KyungTae and Lee, Yong-Hun},
  booktitle={Proceedings of the 31st International Conference on Computational Linguistics: Industry Track},
  pages={760--770},
  year={2025}
}

@inproceedings{karpukhin2020dense,
  title={Dense Passage Retrieval for Open-Domain Question Answering.},
  author={Karpukhin, Vladimir and Oguz, Barlas and Min, Sewon and Lewis, Patrick SH and Wu, Ledell and Edunov, Sergey and Chen, Danqi and Yih, Wen-tau},
  booktitle={EMNLP (1)},
  pages={6769--6781},
  year={2020}
}

@article{lee2023rethinking,
  title={Rethinking the role of token retrieval in multi-vector retrieval},
  author={Lee, Jinhyuk and Dai, Zhuyun and Duddu, Sai Meher Karthik and Lei, Tao and Naim, Iftekhar and Chang, Ming-Wei and Zhao, Vincent},
  journal={Advances in Neural Information Processing Systems},
  volume={36},
  pages={15384--15405},
  year={2023}
}

@inproceedings{scheerer2025warp,
  title={WARP: An efficient engine for multi-vector retrieval},
  author={Scheerer, Jan Luca and Zaharia, Matei and Potts, Christopher and Alonso, Gustavo and Khattab, Omar},
  booktitle={Proceedings of the 48th International ACM SIGIR Conference on Research and Development in Information Retrieval},
  pages={2504--2512},
  year={2025}
}

@inproceedings{li2023slim,
  title={Slim: Sparsified late interaction for multi-vector retrieval with inverted indexes},
  author={Li, Minghan and Lin, Sheng-Chieh and Ma, Xueguang and Lin, Jimmy},
  booktitle={Proceedings of the 46th International ACM SIGIR Conference on Research and Development in Information Retrieval},
  pages={1954--1959},
  year={2023}
}

@article{gao2021coil,
  title={COIL: Revisit exact lexical match in information retrieval with contextualized inverted list},
  author={Gao, Luyu and Dai, Zhuyun and Callan, Jamie},
  journal={arXiv preprint arXiv:2104.07186},
  year={2021}
}

@inproceedings{li2023citadel,
  title={CITADEL: Conditional token interaction via dynamic lexical routing for efficient and effective multi-vector retrieval},
  author={Li, Minghan and Lin, Sheng-Chieh and Oguz, Barlas and Ghoshal, Asish and Lin, Jimmy and Mehdad, Yashar and Yih, Wen-tau and Chen, Xilun},
  booktitle={Proceedings of the 61st Annual Meeting of the Association for Computational Linguistics (Volume 1: Long Papers)},
  pages={11891--11907},
  year={2023}
}

@inproceedings{sivic2003video,
  title={Video Google: A text retrieval approach to object matching in videos},
  author={Sivic and Zisserman},
  booktitle={Proceedings ninth IEEE international conference on computer vision},
  pages={1470--1477},
  year={2003},
  organization={IEEE}
}

@article{jegou2010product,
  title={Product quantization for nearest neighbor search},
  author={Jegou, Herve and Douze, Matthijs and Schmid, Cordelia},
  journal={IEEE transactions on pattern analysis and machine intelligence},
  volume={33},
  number={1},
  pages={117--128},
  year={2010},
  publisher={IEEE}
}

@article{wang2020deltapq,
  title={DeltaPQ: lossless product quantization code compression for high dimensional similarity search},
  author={Wang, Runhui and Deng, Dong},
  journal={{PVLDB}},
  volume={13},
  number={13},
  pages={3603--3616},
  year={2020},
  publisher={VLDB Endowment}
}

@article{wang2021comprehensive,
  title={A comprehensive survey and experimental comparison of graph-based approximate nearest neighbor search},
  author={Wang, Mengzhao and Xu, Xiaoliang and Yue, Qiang and Wang, Yuxiang},
  journal={{PVLDB}},
  volume={14},
  number={11},
  pages={1964--1978},
  year={2021},
  publisher={VLDB Endowment}
}

@article{pan2024survey,
  title={Survey of vector database management systems},
  author={Pan, James Jie and Wang, Jianguo and Li, Guoliang},
  journal={The VLDB Journal},
  volume={33},
  number={5},
  pages={1591--1615},
  year={2024},
  publisher={Springer}
}

@article{malkov2018efficient,
  title={Efficient and robust approximate nearest neighbor search using hierarchical navigable small world graphs},
  author={Malkov, Yu A and Yashunin, Dmitry A},
  journal={IEEE transactions on pattern analysis and machine intelligence},
  volume={42},
  number={4},
  pages={824--836},
  year={2018},
  publisher={IEEE}
}

@inproceedings{indyk1998approximate,
  title={Approximate nearest neighbors: towards removing the curse of dimensionality},
  author={Indyk, Piotr and Motwani, Rajeev},
  booktitle={Proceedings of the thirtieth annual ACM symposium on Theory of computing (STOC)},
  pages={604--613},
  year={1998}
}

@inproceedings{shrestha2024espn,
  title={Espn: Memory-efficient multi-vector information retrieval},
  author={Shrestha, Susav and Reddy, Narasimha and Li, Zongwang},
  booktitle={Proceedings of the 2024 ACM SIGPLAN International Symposium on Memory Management},
  pages={95--107},
  year={2024}
}

@inproceedings{nardini2024efficient,
  title={Efficient multi-vector dense retrieval with bit vectors},
  author={Nardini, Franco Maria and Rulli, Cosimo and Venturini, Rossano},
  booktitle={European Conference on Information Retrieval},
  pages={3--17},
  year={2024},
  organization={Springer}
}

@article{nguyen2016ms,
  title={Ms marco: A human-generated machine reading comprehension dataset},
  author={Nguyen, Tri and Rosenberg, Mir and Song, Xia and Gao, Jianfeng and Tiwary, Saurabh and Majumder, Rangan and Deng, Li},
  year={2016}
}

@misc{colbertviicode,
  author       = {{Stanford Future Data}},
  title        = {{ColBERT}: {O}fficial {G}itHub {R}epository},
  howpublished = {\url{https://github.com/stanford-futuredata/ColBERT}},
  year         = {2020},
  note         = {Accessed: 2026-03-01}
}

@inproceedings{bhat2025xtr,
  title={XTR meets ColBERTv2: Adding ColBERTv2 Optimizations to XTR},
  author={Bhat, Riyaz Ahmad and Sen, Jaydeep},
  booktitle={Proceedings of the 31st International Conference on Computational Linguistics: Industry Track},
  pages={358--365},
  year={2025}
}

@article{DBLP:journals/corr/abs-2312-10997,
  author       = {Yunfan Gao and
                  Yun Xiong and
                  Xinyu Gao and
                  Kangxiang Jia and
                  Jinliu Pan and
                  Yuxi Bi and
                  Yi Dai and
                  Jiawei Sun and
                  Qianyu Guo and
                  Meng Wang and
                  Haofen Wang},
  title        = {Retrieval-Augmented Generation for Large Language Models: {A} Survey},
  journal      = {CoRR},
  volume       = {abs/2312.10997},
  year         = {2023},
}

@inproceedings{wang2021milvus,
  title={Milvus: A purpose-built vector data management system},
  author={Wang, Jianguo and Yi, Xiaomeng and Guo, Rentong and Jin, Hai and Xu, Peng and Li, Shengjun and Wang, Xiangyu and Guo, Xiangzhou and Li, Chengming and Xu, Xiaohai and others},
  booktitle={Proceedings of the 2021 international conference on management of data},
  pages={2614--2627},
  year={2021}
}

@misc{weaviate2025,
  title = {Weaviate},
  url = {http://weaviate.io},
  note = {Accessed: 2026-03-01},
  year = {2026}
}

@misc{qdrant,
  title = {Qdrant},
  url = {https://qdrant.tech/},
  note = {Accessed: 2026-03-01},
  year = {2026}
}

@article{douze2025faiss,
  title={The faiss library},
  author={Douze, Matthijs and Guzhva, Alexandr and Deng, Chengqi and Johnson, Jeff and Szilvasy, Gergely and Mazar{\'e}, Pierre-Emmanuel and Lomeli, Maria and Hosseini, Lucas and J{\'e}gou, Herv{\'e}},
  journal={IEEE Transactions on Big Data},
  year={2025},
  publisher={IEEE}
}

@article{garg2025incorporating,
  title={Incorporating Token Importance in Multi-Vector Retrieval},
  author={Garg, Ankit and Shiragur, Kirankumar and Kayal, Neeraj and others},
  journal={arXiv preprint arXiv:2511.16106},
  year={2025}
}

@book{DBLP:books/mk/HanKP2011,
  author       = {Jiawei Han and
                  Micheline Kamber and
                  Jian Pei},
  title        = {Data Mining: Concepts and Techniques, 3rd edition},
  publisher    = {Morgan Kaufmann},
  year         = {2011},
}

@article{lshapg,
  author       = {Xi Zhao and
                  Yao Tian and
                  Kai Huang and
                  Bolong Zheng and
                  Xiaofang Zhou},
  title        = {Towards Efficient Index Construction and Approximate Nearest Neighbor
                  Search in High-Dimensional Spaces},
  journal      = {{PVLDB}},
  volume       = {16},
  number       = {8},
  pages        = {1979--1991},
  year         = {2023},
}

@book{DBLP:books/daglib/0012859,
  author       = {Michael Mitzenmacher and
                  Eli Upfal},
  title        = {Probability and Computing: Randomized Algorithms and Probabilistic
                  Analysis},
  publisher    = {Cambridge University Press},
  year         = {2005},
}

@article{DBLP:journals/pvldb/GongZC25,
  author       = {Zengyang Gong and
                  Yuxiang Zeng and
                  Lei Chen},
  title        = {Accelerating Approximate Nearest Neighbor Search in Hierarchical Graphs:
                  Efficient Level Navigation with Shortcuts},
  journal      = {{PVLDB}},
  volume       = {18},
  number       = {10},
  pages        = {3518--3530},
  year         = {2025},
}

@misc{mvhnsw-2026,
  author       = {Binhan Yang and Yuxiang Zeng and Hengxin Zhang and Zhuanglin Zheng and Yunzhen Chi and Yongxin Tong and Ke Xu},
  title        = {Unified and Efficient Approach for Multi-Vector Similarity Search (Full Paper)},
  howpublished = {\url{https://github.com/oldherd/MV-HNSW}},
  year         = {2026}
}

@misc{code,
  title        = {Source Code and Dataset},
  howpublished = {\url{https://github.com/oldherd/MV-HNSW}},
  year         = {2026}
}

\end{document}